\documentclass[12pt]{article}

\usepackage{amsfonts}
\usepackage{amssymb}
\usepackage{amsmath}

\def\Journal#1#2#3#4#5#6{#1, #2, {#3} {\bf #4}, (#6) #5.}

\def\CMP{Commun. Math. Phys.}
\def\PRD{Phys. Rev. D}
\def\CQG{Class. Quantum Grav.}
\def\PRL{Phys. Rev. Lett.}
\def\JMP{J. Math. Phys.}

\usepackage{amsthm}

\usepackage{authblk}
\usepackage[titletoc]{appendix}

\addcontentsline{toc}{section}{Appendices}

\newcommand{\bm}[1]{\mbox{\boldmath $#1$}}

\def\X{\mathfrak{X}}
\def\F{{\mathcal F}}

\def\xif{\mbox{\boldmath $ \xi $}}

\def\k{\kappa_\xi}

\def\be{\begin{equation}}
\def\ee{\end{equation}}
\def\bea{\begin{eqnarray}}
\def\eea{\end{eqnarray}}
\def\bean{\begin{eqnarray*}}
\def\eean{\end{eqnarray*}}

\def\={\stackrel{\Sigma}{=}}

\newlength{\cellwidth}
\usepackage{tabularx}
\usepackage{multirow}
\usepackage{array}

\newcolumntype{P}[1]{>{\centering\arraybackslash}p{#1}}
\newcolumntype{M}[1]{>{\centering\arraybackslash}m{#1}}

\def\aa{F}
\def\eqHxi{\stackrel{\H_{\xi}}{=}}
\def\eqHzeta{\stackrel{\H_{\zeta}}{=}}
\def\eqHeta{\stackrel{\H_{\eta}}{=}}
\def\eqHhat{\stackrel{\Hhat}{=}}
\def\Hone{\H_{\xi}}

\def\Hhat{\widehat{\H}}

\def\a{\bm{a}}
\def\b{\bm{b}}
\def\Kill{{\mathcal A}}
\def\eqH{\stackrel{\H_{\xi}}{=}}

\def\ads{\mbox{(A)dS}}
\def\tinyads{\mbox{\tiny (A)dS}}
\def\antids{\mbox{AdS}}
\def\M{\mathbb{M}}
\def\Mink{\M^{1,n}}
\def\Minkpq{\M^{p,q}}
\def\H{{\mathcal H}}

\def\la{\langle}
\def\ra{\rangle}

\newtheorem{theorem}{Theorem}
\newtheorem{definition}{Definition}
\newtheorem{remark}{Remark}
\newtheorem{corollary}{Corollary}

\newtheorem{lemma}{Lemma}

\newcounter{mnotecount}[section]

\renewcommand{\themnotecount}{\thesection.\arabic{mnotecount}}
\newcommand{\mnote}[1]%{}
{\protect{\stepcounter{mnotecount}}$^{\mbox{\footnotesize
$%\!\!\!\!\!\!\,
\bullet$\themnotecount}}$ \marginpar{%\color{red}%
\raggedright\tiny\em
$\!\!\!\!\!\!\,\bullet$\themnotecount: #1} }

\newcommand{\mnotex}[1]%{}
{\protect{\stepcounter{mnotecount}}$^{\mbox{\footnotesize $\bullet$\themnotecount}}$ 
\marginpar{%\color{red}%
\raggedright\tiny\em
$\!\!\!\!\!\!\,\bullet$\themnotecount: #1} }

\begin{document}

\title{Multiple Killing Horizons}
\author[1]{Marc Mars}
\author[2]{Tim-Torben Paetz}
\author[3]{Jos\'e M. M. Senovilla}
\affil[1]{Instituto de F\'isica Fundamental y Matem\'aticas, Universidad de Salamanca, Plaza de la Merced s/n, 37008 Salamanca, Spain}
\affil[2]{Gravitational Physics, University of Vienna, Boltzmanngasse 5, 1090 Vienna, Austria}
\affil[3]{Departamento de F\'isica Te\'orica e Historia de la Ciencia, Universidad del Pa\'is Vasco UPV/EHU, Apartado 644, 48080 Bilbao, Spain}

\maketitle

\vspace{-0.2em}

\begin{abstract}
Killing horizons which can be such for two or more linearly independent Killing vectors are studied. We provide a rigorous definition and then show that the set of Killing vectors sharing a Killing horizon is a Lie algebra $\Kill_\H$ of dimension at most the dimension of the spacetime. We prove that one cannot attach different surface gravities to such multiple Killing horizons, as they have an essentially unique {\em non-zero} surface gravity (or none). $\Kill_\H$ always contains an Abelian (sub)-algebra ---whose elements all have vanishing surface gravity--- of dimension equal to or one less than dim $\Kill_\H$. There arise only two inequivalent possibilities, depending on whether or not there exists the non-zero surface gravity. We show the connection with Near Horizon geometries, and also present a linear system of PDEs, the {\em master equation}, for the proportionality function on the horizon between two Killing vectors of a multiple Killing horizon, with its integrability conditions. We provide explicit examples of all possible types of multiple Killing horizons, as well as a full classification of them in maximally symmetric spacetimes. 
\end{abstract}

%\noindent
%\hspace{2.1em} PACS number:

%\tableofcontents

%$\widehat c$ and $\widehat k$ consistently used?

%different order in section 3

\section{Introduction}

The notion of Killing horizon captures the idea that a 
Killing vector $\xi$ in a spacetime $(M,g)$ may change causal character
precisely on a null hypersurface. In more precise terms, a
Killing  horizon $\mathcal{H}_{\zeta}$  of a Killing vector
$\zeta$ in a spacetime $(M,g)$ is a null hypersurface where $\zeta$
is null, non-zero and tangent. Killing horizons play a fundamental role in general relativity, in particular in the context of black holes in equilibrium: By Hawking's rigidity theorem   the event horizon of a stationary, asymptotically flat black hole spacetime (supplemented by certain additional assumptions, cf.\ \cite{IK} for a review),  is a Killing  horizon.  In fact one often uses 
the notion of Killing horizon to provide a quasi-local definition of
equilibrium black hole. Killing horizons are also relevant particular
cases of more general notions 
such as isolated horizons, weakly isolated horizons, non-expanding horizons
or totally geodesic null hypersurfaces, which have been
extensively studied in the literature (see 
\cite{Haj1, Haj2, Ash1,
Ash1.5, Ash1.7, Ash2, GourgoulhonJaramillo2006, JaramilloPRD,
Lew1, Lew2, Mars2012}
and references therein).
Some physically interesting spacetimes, such as pp-waves, can even be foliated
by Killing horizons.

Now, it can happen that a null hypersurface, or at least a portion thereof, is simultaneously the Killing
horizon of two or more independent Killing vectors.
In fact this is a situation known to happen e.g.\ in Minkowski spacetime where, in standard coordinates,
the null hypersurface $\{ t=x >0\}$ is a Killing horizon of the null translational Killing $\partial_t+ \partial_x$
as well as of the boost in the $x$ direction $x \partial_t + t \partial_x$.
%\tim{also in other cases? ref?}
This article initiates a series of papers where the existence and properties
of these \emph{multiple Killing horizons (MKHs)}
are analyzed in detail.
From a mathematical viewpoint, this problem turns out to be 
remarkably rich and elegant.
% and involves severof view an analysis of this problem yields some nice geometric results. 
%\tim{is there some better  motivation?}\marc{Dropped ``nice'' simply because of its connotations}
Moreover, it leads to
some questions relevant on its own, such as for instance whether near horizon geometries of a multiple degenerate Killing horizon
depend on the Killing vector with respect to which the near horizon limit is performed.

MKHs are also interesting from a physical point of view. 
As mentioned above the event horizon of stationary black holes is a Killing horizon.
Given a   suitably normalized Killing vector $\zeta$ with associated  horizon $\mathcal{H}_{\zeta}$
one can introduce a function $\kappa_{\zeta}$ which provides a measure for  the deviation of the Killing parameter from an affine parameter along the null geodesic generators of $\mathcal{H}_{\zeta}$. 
%\marc{reordering}
Under suitable asymptotics of the spacetime this function is interpreted as the ``surface gravity'' of the black hole, as it 
 determines the redshifted force on a near-horizon test body viewed from infinity \cite{Wald}. It turns out that the surface gravity is constant on
$\mathcal{H}_{\zeta}$ under fairly
general circumstances \cite{Wald}, and this establishes the zeroth law of black
hole thermodynamics. The interpretation of the surface gravity as a 
temperature of the black hole is reinforced by the first and second laws of 
black hole thermodynamics, and turned into a physical certainty
by the Hawking emission process and the corresponding 
Hawking temperature. Thus, when dealing with a
MKH an immediate question arises. To a 
MKH one can ascribe different surface gravities (one for each choice 
of independent Killing vector) and hence also different
temperatures to the black hole. What is the physical meaning of this and what
are its physical consequences? As we will see presently, we find a number of MKH's interesting properties that help in resolving this problem.
%\tim{maybe too much emphasis on this aspect...}\marc{I even insisted... }

In this first paper we  focus on the basic concepts and properties of MKHs.
In Section~\ref{basics} we provide a rigorous definition of a MKH and prove a first  property, namely
that all surface gravities are always constant without any further assumptions.

In Section~\ref{sec:types} we analyze the set of all Killing vectors sharing a null hypersurface $\H$ as MKH, and prove that they constitute a Lie algebra ---denoted by $\Kill_\H$.
% generated by those Killing vectors which share a horizon.
It is further shown  that one merely has to distinguish two cases: either the Lie algebra is Abelian, in which case all Killing vectors are degenerate 
at the horizon (i.e. have vanishing surface gravities), or it is not Abelian, in which case it contains an Abelian subalgebra of co-dimension 1, and one can
find a basis of Killing vectors  such that all  except one are degenerate.
In the first case we call the MKH \emph{fully degenerate}, in the latter one \emph{non-fully degenerate} or just \emph{non-degenerate}. This result states, in particular, that to any MKH one can ascribe a {\it single} non-zero surface
gravity (or temperature) and this is associated to a single Killing 
generator (up to scale, naturally).
Another general property obtained in this section is that,
letting  $n+1$ denote the spacetime dimension, the maximal dimension $m$ of the Lie algebra $\Kill_\H$ is $n$ in the fully degenerate case while it is $n+1$ in the
non-degenerate case.

Section~\ref{sec:Examples} is devoted to explicit examples of spacetimes with MKHs.
In particular we  provide an example which shows that  MKHs with compact cross-sections exist (which might be regarded as particularly relevant from a physical point of view).
Moreover, we  show that  MKHs exist for any $m\in\{2,\dots, n \}$ and $m\in\{2,\dots,n+1\}$ in the fully degenerate and non-degenerate cases, respectively.
In fact, once a spacetime with a fully degenerate MKH has been given  for some  $m\in\{2,\dots, n\}$ an associated  spacetime with non-degenerate MKH is  obtained
by computing its near horizon geometry \cite{KL}. The reason for that is that when performing the near horizon limit an additional Killing vector, which is   non-degenerate, (and possibly others) is added.

Given a spacetime with a MKH $\mathcal{H}$ the various Killing vectors are parallel on $\mathcal{H}$. In Section~\ref{sec:master} we
derive an equation which is satisfied by the proportionality function between two such Killing vectors.
The so-obtained linear PDE system will be called \emph{master equation}.  We also determine its first integrability conditions.

In Section~\ref{App:MaxSym} we provide a complete classification of MKHs for maximally symmetric spacetimes, i.e.\ for Minkowski and (Anti-)de Sitter spacetimes. For the convenience of the reader some details of the 
proof have been  shifted to Appendix~\ref{ApB}. In Appendix \ref{app:1}
we recall (and prove, for completeness) a known property of the zeros of a 
Killing vector.

Let us conclude the introduction with an outlook.
In the subsequent papers  we will face the question raised above concerning the uniqueness 
of  near horizon geometries which arise from a MKH with at least two degenerate Killing vectors.
We  will further analyze the master equation in more detail. Moreover,  we will   construct vacuum spacetimes with MKHs
via characteristic initial value problems. In this case, and
assuming further
that the initial surface is arranged to form a
bifurcate horizon, the master equation evaluated on the bifurcation surface turns out to be not only necessary but also sufficient for
the existence of a MKH in the emerging spacetime.

\subsection{Notation}
$(M,g)$ denotes a connected, oriented and time-oriented $(n+1)$-dimensional Lorentzian manifold with metric $g$ of signature $(-,+,\dots,+)$.
We sometimes call $(M,g)$ the spacetime.
Unless otherwise stated, all submanifolds will be without boundary.
The topological closure of a set $A$ is denoted by $\overline{A}$.
Given a vector (field) $v$ in $TM$, $\bm{v}$ denotes the corresponding one-form, i.e., the metrically
related covector. Similarly, $\omega$ denotes the vector obtained by raising indices of
a one-form $\bm{\omega}$. In general, $\mathfrak{X}(M)$ denotes the set of smooth vector fields on a differentiable manifold $M$.

We will use index-free as well as index notation. Lowercase Greek letters $\alpha,\beta,\dots $ are spacetime indices and run from $0$ to $n$. Small Latin indices $a,b,\dots,h$ are hypersurface indices and take values from $1$ to $n$. Capital Latin indices $A,B,\dots$ are co-dimension-2 submanifold indices running from $2$ to $n$. Finally, small Latin indices $i,j,\dots$ will enumerate the different Killing vectors of multiple Killing horizons and will take values in $\{1,\dots , m\}$, where $m\leq n+1$.

\section{Multiple Killing Horizons: Basics}
\label{basics}
We start by recalling the notion of a Killing Horizon, which will be the basis of the entire paper. %\marc{addition} 
This notion is only relevant when the spacetime dimension
is at least two, which we assume from now on.
\begin{definition}[Killing horizon of a Killing $\xi$]
A smooth null hypersurface $\H_{\xi}$ embedded  in  a spacetime $(M,g)$ is 
a {\bf Killing horizon  of a Killing vector $\xi$} of $(M,g)$ 
if and only if $\xi$ is null on $\H_\xi$, nowhere zero on $\H_\xi$ and tangent to $\H_\xi$.
Killing horizons can have either one or several connected components, but in the latter case we require that the interior of its closure is a smooth connected hypersurface.
\end{definition}
The reason to allow for multiple connected components will become clear later, as this is needed in our main definition \ref{def:MKH}, and will be illustrated in the examples of section \ref{sec:Examples}.

A more general notion is that of a {\bf Killing prehorizon}. Its definition 
is the same as for Killing horizon except that 
the condition that $\H_{\xi}$ is embedded is replaced by injectively immersed. We will also need the related concept of bifurcation at Killing horizons \cite{Boyer,KW,RW}.
\begin{definition}[Bifurcate Killing horizon]
Let $\xi$ be a Killing vector on $(M,g)$ which has a connected and spacelike co-dimension two submanifold $S$ of fixed points (i.e, such that $\xi|_S =0$). Then, the 
%\marc{changed to be precise}
set of points along all null geodesics orthogonal to $S$ comprises what is called a {\bf bifurcate Killing horizon} with respect to $\xi$.
\end{definition}
%\marc{Changed, as $\H_1$ and $\H_2$ are not Killing horizons according to our  definition}

Observe that the null geodesics orthogonal to $S$ generate two transversal null hypersurfaces $\H_1$ and $\H_2$.
% which are Killing horizons with respect to $\xi$ \cite{KW}. 
The portions $\H_1^+$ and $\H_2^+$ to the future of $S$, as well as the portions $\H_1^-$ and $\H_2^-$ to its past, are all connected Killing horizons. Moreover, $\H_1^+ \cup \H_1^- \subset \H_1$ is also a Killing horizon according
to our definition (since its closure is $\H_1$, which is 
open and connected). The same holds for 
$\H_2^+ \cup \H_2^-$.  Note that $\H_1$, $\H_2$ are {\it not}
Killing horizons.
%connected components of the Killing horizons $\H_1$ and $\H_2$, respectively.
 The union $\H_1^+\cup \H_2^+\cup\H_1^-\cup \H_2^-\cup S = \H_1\cup \H_2$ is the bifurcate Killing horizon.

Our main goal is the study of the following particular class of Killing Horizons.
\begin{definition}[Multiple Killing horizon (MKH)]\label{def:MKH}
A null hypersurface $\H$ embedded  in a spacetime $(M,g)$ is 
a {\bf multiple Killing horizon} if
 $(M,g)$ admits 
Killing horizons  $\H_{\xi_i}$, $i\in\{1,\dots ,m\}$ with $m\geq 2$,
associated to linearly independent Killing vectors
$\xi_i$ satisfying
\begin{align*}
\overline{\H} = \overline{\H}_{\xi_1} = \dots = \overline{\H}_{\xi_m}.
\end{align*}
\end{definition}
Note that if $\H$ is a MKH, so
it is any open subset of $\H$ whose closure is connected. Observe also that 
any hypersurface  containing $\H$ and contained in $\overline{\H}$
is also a MKH. This stems
from the fact that the definition involves $\overline{\H}$.
The reason behind taking this closure in the definition is that
it is not generally true that, say,
$\H_{\xi_1} = \H_{\xi_2}$ and only their closures agree (see section \ref{sec:Examples} for some illustrative examples).
Nonetheless, the case when $\H_{\xi_i} = \H_{\xi_j}$ for all $i,j\in\{1,\dots, m\}$ seems to be still feasible, though it is much rarer.
%\jose{Here, I have commented out the definition of "strict" MKH}
%\marc{ok with me}
%so it deserves a name. Since open subsets of Killing horizons are still Killing horizons, we need to add
%a maximality condition on $\H_{\xi_1}$ and $\H_{\xi_2}$ for the equality
%$\H_{\xi_1} = \H_{\xi_2}$ to be of relevance. In precise terms:
%%\marc{Provisional name}
%\begin{definition}[Strict MKH]
%A multiple Killing horizon $\H$ is called {\bf strict}
%if there exist two Killing horizons $\H_{\xi_1}$, $\H_{\xi_2}$ of 
%linearly independent Killing vectors $\xi_1$, $\xi_2$ for which
%\begin{align*}
%\H = \H_{\xi_1} = \H_{\xi_2}
%\end{align*}
%and there is no Killing horizon of $\xi_1$ or $\xi_2$ strictly containing
%$\H_{\xi_1}$ or $\H_{\xi_2}$.
%\end{definition}
%
%\begin{remark}
%Note that a strict multiple Killing horizon $\H$ may a priori be a multiple
%horizon with respect to  another pair of Killings 
%$\xi_3$, $\xi_4$ with respect to which $\H$ is not strict.
%\end{remark}
%

Killing (pre)horizons of a Killing vector $\xi$ have an associated notion of {\bf surface gravity}, which is a smooth function $\kappa_{\xi} : \H_{\xi} \longrightarrow\mathbb{R}$ defined by
\begin{align}
\nabla_{\xi} \xi \eqH  \kappa_{\xi} \xi \quad 
\quad \mbox{or equivalently} \quad \quad \mbox{grad} (g(\xi,\xi)) \eqH  -2 \kappa_{\xi} \xi
\label{SurGrav}.
\end{align}
If this function vanishes, then $\H_\xi$ is said to be {\em degenerate}. It is very easy to check that $\k$ is constant along the null generators of $\H_\xi$, that is
\be
\xi (\k) =0 .
\ee
One can show that $\k$ has the following useful representation \cite{Wald,FN} (a justification will be provided later in section \ref{sec:master})
\be
\k^2 \eqH -\frac{1}{2} \nabla_\mu\xi_\nu \nabla^\mu \xi^\nu  \label{kappasquare}
\ee
which allows us to prove that $\k$ actually extends as a smooth function to the whole connected 
%\marc{``hypersurface'' deleted, as $\H_{\xi}$ may have boundary, hence would not be a hypersurface as defined before} 
$\overline\H_\xi$, despite the fact that $\H_\xi$ may have several connected components.

We are going to prove that, actually, for any MKH all possible surface gravities are constant. To that end, we need an intermediate basic result. Let $\H$ be a MKH with respect to the Killing vectors $\xi$ and $\eta$. 
Set 
$$\Hhat := \H_{\xi} \cap \H_{\eta}$$
and let $\aa : \Hhat \longrightarrow \mathbb{R}$ be the scalar function defined by
\be
\eta \eqHhat \aa \xi \label{etaxi} .
\ee
By construction $\aa$ is well-defined, smooth and nowhere zero. This function extends
smoothly (and uniquely) to all $\H_{\xi}$ but the extension may have zeroes. Furthermore, $F$ cannot be constant on any open subset ${\cal U} \subset \H$. This follows from the fact that the set of fixed points of a Killing vector cannot have co-dimension one (this is known, but we include a proof in Appendix \ref{app:1}) and the Killing vector $\eta -F_0 \xi$ would vanish on ${\cal U}$ if $F |_{\cal U} = \mbox{const} := F_0$.

\begin{lemma}\label{eq:xiF}
Let $\H$ be a MKH with respect to the Killing vectors $\xi$ and $\eta$ and denote by $\kappa_{\xi}$ and $\kappa_{\eta}$ the surface gravities of $\xi$ on $\H_\xi$ and $\eta$ on $\H_\eta$, respectively. Then 
%\marc{small change in the equal sign}
\begin{align}
\kappa_{\eta} \eqHxi \xi (\aa) + \aa \kappa_{\xi}. 
\label{PDEaa}
\end{align}
\end{lemma}
\begin{proof}
A direct calculation using (\ref{etaxi}) provides 
$$
\kappa_\eta \eta \eqH \nabla_{\eta} \eta  \eqH \left( 
F^2 \k +F\nabla_\xi F \right)\xi 
$$
from where we deduce
$$
\xi (F)+\k F \eqH \kappa_\eta  \label{kappai}
$$
which holds even at the fixed points of $\eta$ (where $F$ vanishes), because the set of fixed points of $\eta$ can have at most co-dimension 2, and thus it follows by continuity.
\end{proof}

As mentioned above, the surface gravities are constant along the null generators, so
the PDE (\ref{PDEaa}) can be explicitly integrated. 
Let $\tau : \H_{\xi} \cap \H_{\eta}  \longrightarrow \mathbb{R}$
be a (smooth) scalar function satisfying $\xi (\tau ) = 1$. Obviously $\tau$ is not univocally defined, as it is affected by the freedom:
\be
\tau \longrightarrow \tau+\tau_0, \hspace{1cm} \xi (\tau_0)=0. \label{ufreedom}
\ee
This freedom can be fixed by giving initial data on any cut $S_0\subset \H$ transversal to $\xi$ but, for the time being, we leave this free.
Now define  $Q_{\xi} : \Hone 
\longrightarrow \mathbb{R}$ by
\begin{align}
Q_{\xi} := \left \{ 
\begin{array}{ll} 
- \frac{1}{\kappa_{\xi}} \left ( e^{-\kappa_{\xi} \tau} - 1 \right ) 
\quad \quad & \mbox{ if } \kappa_{\xi} \neq 0,  \\
 \tau  \quad \quad & \mbox{ if } \kappa_{\xi} = 0.  
\end{array}
\right . \label{Qxi}
\end{align}
This is a smooth function on $\H_{\xi}$ irrespectively of whether
$\kappa_{\xi}$ has zeros or not. Note also that
$\xi (Q_{\xi} ) = e^{-\kappa_{\xi} \tau}$. Then, the general solution of (\ref{PDEaa}) is given in terms of a smooth nowhere zero function  $f : \H_{\xi} \cap \H_{\eta}
\longrightarrow \mathbb{R}$ satisfying $\xi (f )=0$, by 
%\marc{$\xi(f)$ deleted from (\ref{expaa}) to avoid repetition}
\begin{align} 
\aa = f e^{-\kappa_{\xi} \tau} + \kappa_{\eta} Q_{\xi}.
%\hspace{1cm} \xi(f) =0 .
\label{expaa}
\end{align}
Indeed
\begin{align*}
\xi (\aa) + \aa \kappa_{\xi} 
&= \xi \left ( f e^{-\kappa_{\xi} \tau} + \kappa_{\eta} Q_{\xi} \right )
 + \kappa_{\xi} \left ( f e^{-\kappa_{\xi} \tau} + \kappa_{\eta} Q_{\xi}  \right ) \\
& =  - \kappa_{\xi} f e^{-\kappa_{\xi} \tau } + \kappa_{\eta} e^{-\kappa_{\xi} \tau} 
+ \kappa_{\xi} \left ( f e^{-\kappa_{\xi} \tau} + \kappa_{\eta} Q_{\xi}  \right )
= \kappa_{\eta} \left ( e^{-\kappa_{\xi} \tau}  + \kappa_{\xi} Q_{\xi} \right )
= \kappa_{\eta}. 
\end{align*}
As before $f$ extends smoothly to $\H_{\xi}$, possibly with zeroes. 

We can now prove that in MKHs, all the surface
gravities are necessarily constant. 
\begin{theorem}
Let $\H$ be a multiple Killing horizon and $\H_{\xi}$, $\H_{\eta}$
be Killing horizons satisfying $\overline{\H}_{\xi} =
\overline{\H}_{\eta} = \overline{\H}$. Then the respective surface gravities
$\kappa_{\xi}$ and $\kappa_{\eta}$ are constant.
\end{theorem}

\begin{remark}\label{constantSurface}
Constancy of the surface gravity is known to hold in several circumstances,
namely when the Killing generator is integrable \cite{RW} (i.e.
$\bm{\xi} \wedge d \bm{\xi} =0$)), or when the Einstein
tensor of $(M,g)$ satisfies the dominant energy condition \cite[Chapter 12]{Wald}, or for bifurcate Killing horizons \cite{KW,FN}. 
For multiple Killing horizons the constancy of the surface gravity turns
out to be a universal property.
\end{remark}

\begin{proof}
In the multiple horizon case
we work on $\Hhat := \H_{\xi} \cap \H_{\eta}$. Since
$\overline{\H_{\xi} \cap \H_{\eta}} = \overline{\H_{\xi}}
= \overline{\H_{\eta}} = \overline{\H}$, proving constancy on this
set also proves it in the respective Killing horizons.  

Any Killing horizon has a vanishing second fundamental form relative to the one-form $\xif$, as follows from the fact that\footnote{We use the same symbol
$X$ to denote a vector field $X \in \X(\Hhat)$ and its image in 
$T_{\Hhat} M$ under the embedding from $\H$ into $M$.
The precise meaning will be clear
from the context.}
%\tim{is  $\mathfrak{X}$ standard notation, I think this space is not defined (it's only mentioned in the middle of page 5 what it means)}
$$
g(X,\nabla_X \xi)\eqHhat 0 , \hspace{1cm} \forall X\in \mathfrak{X} (\Hhat)
$$
if $\xi$ is a Killing vector. 
 This implies the existence of a one-form $\bm{\varphi}\in \Lambda(\Hhat)$ such that
\be
\nabla_X \xi \eqHhat \bm{\varphi} (X) \xi , \hspace{1cm} \forall X\in \mathfrak{X} (\Hhat). \label{varphi} 
\ee
Taking the covariant derivative along $X$ of the first in (\ref{SurGrav}) and using (\ref{varphi}) 
$$
\bm{\varphi} (X)\xi^\rho \nabla_\rho \xi^\mu +\xi^\rho X^\sigma\nabla_\sigma\nabla_\rho \xi^\mu \eqHhat X(\k) \xi^\mu +\k \bm{\varphi} (X) \xi^\mu .
$$
But any Killing vector satisfies \cite{Wald}
\be
\nabla_\sigma\nabla_\rho \xi_\mu =\xi_\nu R^\nu{}_{\sigma\rho\mu} \label{D2xi}
\ee
where $R^\nu{}_{\sigma\rho\mu} $ is the Riemann tensor of $(M,g)$, so that using (\ref{SurGrav}) again in the previous expression we arrive at
\be
\xi^\rho X^\sigma R^\nu{}_{\sigma\rho\mu}\xi_\nu  \eqHhat X(\k) \xi_ \mu \label{Xxi} .
\ee
The same calculation for $\eta$ leads to
$$
\eta^\rho X^\sigma R^\nu{}_{\sigma\rho\mu}\eta_\nu  \eqHhat X(\kappa_\eta) \eta_ \mu 
$$ so that using here (\ref{etaxi}) and combining with (\ref{Xxi}) we get
\begin{align*}
X (\kappa_{\eta} ) \eqHhat \aa X(\kappa_{\xi}), \hspace{1cm} \forall X\in \mathfrak{X}(\Hhat)
\end{align*}
where $X$ is any vector field tangent to $\Hhat$. Now, the combination of this with (\ref{kappai}) gives the desired result, as $F$ given in (\ref{expaa}) has $\tau$-dependence while the surface gravities do not. To be precise, choose any $X \in \X(\Hhat)$ such that $[\xi ,X]=0$ and take the directional derivative along $\xi$ of the previous expression
$$
\xi (X (\kappa_{\eta} )) \eqHhat \xi(\aa) X(\kappa_{\xi}) +\aa \xi (X(\kappa_\xi)) \hspace{5mm} \Longrightarrow \hspace{3mm}
X (\xi (\kappa_\eta)) \eqHhat \xi(\aa) X(\kappa_{\xi}) + \aa X(\xi(\kappa_\xi)) 
$$
and now use (\ref{kappai}) and $\xi (\kappa_\eta ) =F^{-1} \eta (\kappa_\eta) =0$ to get
$$
\xi(\aa) X(\k)=0
$$
which holds for arbitrary $X\in \mathfrak{X}(\Hhat)$ as long as it commutes with $\xi$. 
If $X (\kappa_{\xi} ) \neq 0$ on some open, connected
and non-empty subset ${\cal U}\subset \Hhat$, then $\xi(\aa)\stackrel{{\cal U}}{=}0$ would necessarily follow, so that from (\ref{PDEaa}) $\kappa_{\eta} = \aa \kappa_{\xi}$ would hold on ${\cal U}$ . By restricting  ${\cal U}$ 
if necessary
%\marc{slight rephrasing, since ``might'' denotes also  remote possibility}  
we would then have that $\kappa_{\xi}$ vanishes nowhere in this set, and consequently
\begin{align*}
& \aa X(\kappa_{\xi} ) = 
X (\kappa_{\eta} ) \stackrel{{\cal U}}{=} 
X  ( \aa \kappa_{\xi}) \stackrel{{\cal U}}{=} 
X (\aa) \kappa_{\xi}
 + \aa X(\kappa_{\xi})  \\
&  \Longrightarrow \quad \quad
X(\aa) \kappa_{\xi} \stackrel{{\cal U}}{=}
0 \hspace{2cm}
\quad \quad \Longrightarrow \quad \quad  X(\aa) 
\stackrel{{\cal U}}{=} 0
\end{align*}
implying that $\aa $ would be a constant on ${\cal U}$, say $\aa_0$. But then
the Killing vector $\eta - \aa_0 \xi$ would vanish on a co-dimension one
subset of the spacetime, hence everywhere, and $\eta$ would not be linearly independent of $\xi$,
against hypothesis.

Hence, $X(\kappa_{\xi})$ must vanish on a dense subset of $\Hhat$ ---for arbitrary $X$ subject to $[\xi,X]=0$---, then also $X(\kappa_{\eta})$ vanishes there, and both $\kappa_{\xi}$ and
$\kappa_{\eta}$ are constant on any connected component of
$\H_{\xi} \cap  \H_{\eta}$.% \marc{phrase in parenthesis deleted, as not our KH are no longer defined to be connected}
By continuity of $\kappa_{\xi}$ on $\H_{\xi}$ it follows that this surface
gravity is constant on $\overline \H$ and the same holds for $\kappa_{\eta}$. 
\end{proof}

\section{Multiple Killing Horizons: Lie algebra and types}\label{sec:types}
In this section, we start by proving that the set of all Killing vectors in $(M,g)$ with a common multiple Killing horizon constitute a Lie sub-algebra of the Killing Lie algebra, and we also find the possible structure constants and dimensions. This will then allow for distinguishing between different types of MKHs, which will be rigorously defined.

For any spacetime $(M,g)$ we denote by $\Kill_M$ the Lie algebra of Killing
vectors. This is a finite dimensional vector space of dimension bounded above 
by $(n+1)(n+2)/2$. Consider a multiple Killing horizon $\H$ and define
$\Kill_\H$
as the  union of the trivial Killing 
vector  and the collection of Killing vectors $\xi$ which admit
a Killing horizon $\H_{\xi}$ satisfying $\overline{\H} = \overline{\H_{\xi}}$.
% ---including the trivial Killing vector field.
%Define also $\Kill_\H := \Kill^0_\H \cup \{ \bm{0}\}$ where $\bm{0}$ is the
%trivial Killing field.
It turns out that $\Kill_{\H}$ is a Lie sub-algebra of 
$\Kill_{M}$.  
\begin{theorem}
\label{KillH}
Let $\H$ be a multiple Killing horizon in a spacetime 
$(M,g)$ of arbitrary dimension at least two. Then $\Kill_{\H}$ is a Lie sub-algebra
of the Killing algebra $\Kill_{M}$ of $(M,g)$. 
\end{theorem}

\begin{proof}
First we prove that $\Kill_\H$ is a vector sub-space of $\Kill_M$. 
Let $\xi, \eta \in \Kill_\H$.  We want to show that $a_1 \xi + 
a_2\eta \in \Kill_\H$, for any
$a_1, a_2 \in \mathbb{R}$. If either $\xi$ or $\eta$ is the zero vector, the claim
is obvious. Assume both $\xi$ and $\eta$ are non-trivial. Then
there exists a hypersurface $\widehat{\H}$ 
which is a Killing horizon with respect to both $\xi$ 
and $\eta$ and $\widehat{\H}$ is a dense subset of $\H$. We know (\ref{etaxi})  that $\xi$ and
$\eta$ are proportional (and null) on $\widehat{\H}$, so the  Killing vector
 $\zeta:= a_1 \xi + a_2 \eta$
is also tangent to $\widehat{\H}$ and null there. Moreover, if it vanishes on a dense subset
of $\widehat{\H}$, then by Lemma \ref{codimension_two} in  Appendix 
\ref{app:1}, it vanishes identically, hence
belongs to $\Kill_\H$. Otherwise, there exists an open and dense
$\H_{\zeta}  \subset \widehat{\H}$ where $\zeta$ does not vanish. In other
words, $\H_{\zeta}$ is a Killing horizon of $\zeta$. Given
that $\overline{\H_{\zeta}} = \H$, we conclude 
$\zeta \in \Kill_\H$, as claimed.

It remains to prove that the commutator of any two Killing vectors $\xi,\eta\in \Kill_\H$ also belongs to $\Kill_\H$. Of course, $[\xi,\eta]\in \Kill_M$ for arbitrary $\xi,\eta\in \Kill_\H$, so we only need to show that $[\xi,\eta]$ is null and tangent to (a dense subset of) $\overline\H$.
But we know that the Killing vectors $\xi$ and $\eta$ are related by (\ref{etaxi}). Given also that they are tangent to $\H$, we can compute their commutator there
\be
[\xi,\eta]\eqHxi [\xi,F\xi]\eqHxi \xi(F) \xi \eqHxi (\kappa_\eta - F \k) \xi \label{comm1}
\ee
where in the last step we have used (\ref{PDEaa}). This finishes the proof.
\end{proof}
\begin{definition}[Lie algebra and order of a MKH]
%\jose{I wonder if ``multiplicity'' would be a better name, rather than ``order''}
%\marc{Either is fine with me``multiplicity''. For the moment, kept as it was...} 
%\tim{if you prefer ``multiplicity'' it's fine with me}
We call $\Kill_\H$ the Lie algebra of the multiple Killing horizon $\H$.

The {\bf order} $m\geq 2$ of a MKH $\H$ is, by definition, the dimension of its Lie algebra $\Kill_\H$.
\end{definition}
We shall sometimes loosely speak of double, triple, quadruple, etcetera, MKHs for $m=2,3,4, \dots$.

We can actually say much more about $\Kill_\H$ and its order.

\begin{theorem}\label{th:algebra}
Let $\Kill_{\H}$ by the Lie algebra of a MKH $\H$ of order $m$  in a spacetime 
$(M,g)$ of arbitrary dimension at least two. Then, $\Kill_\H$ always contains an Abelian sub-algebra $\Kill_\H^{deg}$ of dimension at least $m-1$ whose elements have vanishing surface gravities, that is to say, they all have (the appropriate dense subset of) $\H$ as a {\em degenerate} Killing horizon. If this Abelian sub-algebra $\Kill_\H^{deg}$ has dimension $m-1$, the remaining independent Killing vector (say $\xi$) in $\Kill_\H\setminus \Kill_\H^{deg}$ has $\k\neq 0$ and satisfies
\be
\left[\xi,\eta \right] =-\k \eta , \hspace{1cm} \forall \eta \in \Kill_\H^{deg}.
\ee
\end{theorem}
\begin{proof}
Let $\{\eta_i\}$ be a basis of $\Kill_\H$, and let $\zeta \in \Kill_\H$ be non-trivial, otherwise arbitrary. Then
$$
\zeta = b^i \eta_i
$$
where $b^i\in \mathbb{R}$ are constants. 
%\marc{A small addition and a few changes}
Fix a non-zero element
$\xi \in \Kill_H$ and let $\H_{\xi} \subset \H$ be its corresponding
Killing horizon.  Expression (\ref{etaxi}) holds for each $\eta_i$ with 
corresponding functions $F_i$. 
From the definition of surface gravity (\ref{SurGrav})
the acceleration of $\zeta$ on $\H_{\xi}$ is
$$
\nabla_\zeta \zeta =b^i b^j \nabla_{\eta_i}  \eta_j  \eqHxi b^i b^j F_i \nabla_\xi (F_j \xi)\eqHxi (b^i F_i) b^j (F_j \k + \xi(F_j))\xi \eqHxi (b^j \kappa_{\eta_j} ) b^i \eta_i \eqHxi  (b^j \kappa_{\eta_j}) \zeta 
$$
where in the penultimate step we have used the PDE (\ref{PDEaa}) for the functions $F_j$. This proves that the surface gravity of $\zeta$ on $\H$ is 
$$
\kappa_\zeta = b^j \kappa_{\eta_j} .
$$
It follows that every $\zeta\in \Kill_\H$ with 
\be
b^j  \kappa_{\eta_j} =0 \label{deg}
\ee
has a vanishing surface gravity. There are at least $m-1$ linearly independent such degenerate Killing vectors, as follows from the following elementary reasoning: the relation (\ref{deg}) can be seen as the scalar product of the constant vectors $(b^j)$ and $(\kappa_{\eta_j})$ on an $m$-dimensional vector space, so that given $(\kappa_{\eta_j})$ as data, there exist $m-1$ linearly independent solutions for $(b^j)$ ---if at least one of the $\kappa_{\eta_j}$ does not vanish. If all the surfaces gravities $\kappa_{\eta_j}$ vanish then every $\zeta \in \Kill_\H$ has vanishing $\kappa_\zeta$ too. \footnote{To avoid cumbersome notation
we define the surface gravity of the zero vector to be zero.}

To end the proof, we use (\ref{comm1}). For, if $\xi$ and $ \eta$ both have vanishing surface gravity, then (\ref{comm1}) informs us that $[\xi,\eta] \eqHxi 0$ and therefore the Killing vector $[\xi,\eta]$ must vanish everywhere. This proves that $\Kill_\H^{deg}$ is Abelian. Similarly, if only $\eta$ has $\kappa_\eta =0$, then (\ref{comm1}) implies that $[\xi,\eta] \eqHxi -\k \eta$, and thus the Killing vector $[\xi,\eta] +\k \eta$ must vanish everywhere, finishing the proof.
\end{proof}
\begin{remark}[Notation]\label{notation}
In summary we have proven that, for multiple Killing horizons of order $m$, there is always a basis of $\Kill_\H$ with $m-1$ degenerate Killings vectors all of them commuting. Therefore, from now on we will use the following useful notation: $\{\eta_i\}$, $i=1,\dots, m$, will always denote a basis of $\Kill_\H$ with $\{\eta_2,\dots,\eta_m\}$ a basis of $\Kill_\H^{deg}$, that is to say,
$$
\kappa_{\eta_2} = \dots = \kappa_{\eta_m} =0.
$$
Then, we will also use the name $\xi =\eta_1$, and $\k$ is arbitrary (it may vanish or not). 
\end{remark}
With this choice of basis we have found all the structure constants of $\Kill_\H$: 
%\marc{small change in the values of $j,k$}
$$
C^i_{jk}=0 , \quad \quad C^j_{1k} =-\k \delta^j_k, \quad \forall k\neq 1 .
$$

\begin{definition}[Fully degenerate MKH]
A multiple Killing horizon $\H$ is said to be {\bf fully degenerate} if $\Kill_\H = \Kill_\H^{deg}$, that is to say, if its Lie algebra is Abelian, and all surface gravities vanish.
\end{definition}
Observe that non-fully degenerate MKHs possess a {\em unique} non-zero surface gravity. To fix the value of this surface gravity requires the use of some normalization for the Killing vector $\xi$, be it at infinity or in some other appropriate place. This has some physical implications, as one cannot attach two different non-zero surface gravities to a given MKH, despite the fact of being a Killing horizon for multiple Killing vectors.

\begin{corollary}\label{coro:dim}
The maximum possible dimension of $\Kill_\H^{deg}$ is $n=$ {\em dim}$(M)-1$. Therefore, the maximum possible order of a MKH $\H$ is
\begin{enumerate}
\item $m=n$ for fully degenerate $\H$,
\item $m=n+1$ for non-fully degenerate $\H$.
\end{enumerate}
\end{corollary}
\begin{proof}
As $\Kill_\H^{deg}$ is Abelian, its dimension can be at most $n+1$. But if it were $n+1$ the spacetime would be homogeneous, and actually locally flat (this follows from the fact that the Riemann tensor on the orbits of a group of motions can be expressed in terms of the structure constants of its Lie algebra, and it vanishes for Abelian groups \cite{Exact}), in a neighbourhood around $\H$, and this is not possible, as the Abelian sub-algebra is generated by
translations, and hence its span is $n+1$ dimensional at every point.
Thus, dim$(\Kill_\H^{deg})$ is at most $n$.
\end{proof}
The bound $m \leq n+1$ is sharp. Examples where
the maximal value $m=n+1$ is attained are the maximally symmetric spacetimes $(M,g)$, see section \ref{sec:Examples} for explicit examples
and Section \ref{App:MaxSym}, where we present the full classification of MKHs in maximally symmetric spacetimes.

Using the notation fixed in Remark \ref{notation}, the expressions (\ref{expaa}) for the elements $\eta_i\in \Kill_\H^{deg}$ then reduce simply to
\be
F_i =f_i e^{-\k \tau} , \hspace{1cm} \xi(f_i) =0 \hspace{1cm} \forall i\in\{2,\dots , m\} \label{As}
\ee
valid for both cases with $\k$ zero or not. Then we have the relations
\be
\eta_i \eqH f_i e^{-\k \tau} \xi , \hspace{1cm} \xi(f_i) =0 \hspace{1cm} \forall i\in\{2,\dots , m\} .\label{etaxi1}
\ee
The freedom (\ref{ufreedom}) translates to a simple redefinition $f_i\rightarrow f_i e^{-\k \tau_0}$ which is consistent given that $\xi (\tau_0)=0$. Note that the zeros of the functions $f_i$ are fixed points of the corresponding Killing vectors. These fixed points of each $\eta_i$ are not part of the Killing horizons $\H_{\eta_i}$, but they do belong to their closure and thus to $\overline\H$. 

Given that $\kappa_{\eta_i}=0$ for all $i\in\{2,\dots,m\}$, the vector fields $\eta_i$ have zero acceleration on 
their corresponding horizons $\H_{\eta_i} \subset \H$,
and thus their integral curves are affinely parametrized null geodesics generating $\H_{\eta_i}$. Then, the relations (\ref{etaxi1}) imply that an affine parameter $\lambda_i$ along the geodesics tangent to $\eta_i$ in
$\H_{\eta_i} \cap \H_{\xi}$
 are given, for the non-fully degenerate case $\k\neq 0$, by
$$
\lambda_i = \frac{1}{\k f_i} e^{\k \tau} , \hspace{1cm} \forall i\in\{2,\dots,m\}
$$
and therefore, the integral curves of $\xi$ in
$\H_{\xi} \cap \H_{\eta_i}$ are incomplete geodesics (the range of the affine parameter $\lambda_i$ cannot be the whole real line).%\jose{Changes next, to see if we are all happy}
%\marc{small change, because $\tau$ may not extend to $- \infty$}
%   with $\lambda_i =0$ when $v\rightarrow -\infty$. 

%\marc{I find the following paragraph somewhat imprecise. The spacetime
%is something given  so taking about extensions of geodesics without taking about
%extensions of spacetimes is vague. And extensions of spacetimes are delicate
%by themselves. For instance the extension 
%proven by Racz and Wald is not global; it is local in a neighbourhood
%on a given null generator...}
%\tim{
%I'm not too familiar with this result, but in the paper they do some globalization procedure, so it seems that they extend a neighbourhood of the horizon...
%}
Using the results in \cite{Boyer}, see also \cite{FN,RW},
%\tim{ref \cite{RW} is mentioned twice}
and as $\k\neq 0$ is constant, we deduce that these incomplete geodesics do not reach any curvature singularity, and therefore they are only a segment of a larger geodesic in the given spacetime, or the latter is extendable. Actually, the integral curves of $\eta_i$ are longer geodesics if the given spacetime contains them ---otherwise, they could be extended in any proper extension of the spacetime--- and along them $\xi$ vanishes on a co-dimension two null
submanifold $S\subset \overline \H$. Therefore, %any \tim{``any'' sound like as if no technical conditions are needed for this?}
 non-fully degenerate multiple Killing horizons can be seen as a branch of a bifurcate Killing horizon with $\xi$ as the bifurcate Killing vector field and $S:= \{\xi =0\} \cap \overline\H$ as the bifurcation surface. 

\section{Examples}\label{sec:Examples}
In this section we present explicit examples of MKHs with the aim of illustrating the previous results and to gain some insight on their structure. We will also show that all possible types of MKHs exist, fully degenerate or not, and of any possible admissible order.

\subsection{Flat spacetime}\label{Minkowski}
In $(n+1)$-dimensional Minkowski spacetime $(\mathbb{R}^{n+1}, g^{\flat})$, where $g^\flat$ is the flat metric (with vanishing curvature tensor), any null hyperplane is a MKH of maximal order $m=n+1$ (and therefore, non-fully degenerate). To check this, choose a global Cartesian coordinate system $\{t,x^a\}$ such that
\be
g^\flat = -dt^2 + \sum_{a=1}^n (dx^a)^2, \label{gflat}
\ee
and select, for instance, the null hyperplane $\H := \{t=x^1\}$. Let $A\in \{2,\dots,n\}$ and consider the following collection of $n+1$ linearly independent Killing vectors of $(\mathbb{R}^{n+1}, g^{\flat})$:
\begin{align}
\eta_1 &= \xi = x^1\partial_t + t \partial_{x^1}, \label{eta1}\\
\eta_2 &= \partial_t +\partial_{x^1}, \label{eta2}\\
\eta_{A+1} &=x^A\partial_t+t\partial_{x^A} + x^A\partial_{x^1} -x^1\partial_{x^A}  =x^A\left( \partial_t +\partial_{x^1}\right) +(t-x^1) \partial_{x^A}  \label{etaA}.
\end{align}
These are all obviously null, and proportional to $\eta_2$, at $\H$. $\eta_2$ is non-zero everywhere, and thus the entire $\H$ is a Killing horizon for $\eta_2$. On $\H$ we also have $\eta_{A+1}\stackrel{\H}{=} x^A\left( \partial_t +\partial_{x^1}\right)$,
%\tim{I guess this should be $\eta_{A+1}$}
 so that each $\eta_{A+1}$ vanishes on the co-dimension two surface $\H\cap \{x^A=0\}$. Thus the Killing horizon for each $\eta_{A+1}$ is given by $\H_{\eta_{A+1}}=\H\setminus \{x^A=0\}$, has two connected components given by $x^A>0$ and $x^A<0$, but also $\overline{\H}_{\eta_{A+1}}=\H =\overline\H$.
%\tim{the same} 
Concerning $\eta_1=\xi$, we have 
$\xi\stackrel{\H}{=} t\left( \partial_t +\partial_{x^1}\right)$ and thus $S:=\H\cap \{t=0\} =\{t=x=0\}$ is a co-dimension two spacelike surface of fixed points for $\xi$. The Killing horizon $\H_\xi$ of $\xi$ has thus two connected components defined by $t>0$, say $\H_1^+$, and by $t<0$, say $\H_1^-$, but again $\overline{\H}_\xi =\H$. Therefore, $\H$ is a multiple Killing horizon of maximal order $m=n+1$. 

All the Killing vectors shown above except $\eta_1=\xi$ are affinely parametrized geodesic vector fields on $\H$, and thus their surface gravities vanish. Also, $\nabla_\xi \xi =\xi$ so that $\k =1$.

Observe that $\H_\xi$, $\xi$ having a set of fixed points at $S$, is a branch of a bifurcate Killing horizon, the second branch being given by the hyperplane $\{t+x^1=0\}$ which provides the future and past connected components $\H_2^+$ and $\H_2^-$ for $t>0$ and $t<0$, respectively. This hyperplane is itself a MKH of maximal order. 

The full classification of MKHs in flat spacetime, as well as (anti)-de Sitter spacetimes, is presented in Section \ref{App:MaxSym}.

\subsection{A double Killing horizon with compact sections}
Consider the two-dimensional de Sitter space $dS_2$ of constant curvature
$\varkappa^2$ and the two-dimensional sphere $\mathbb{S}^2$ with the 
round metric of radius $1/\varkappa$. 
The Nariai spacetime is the product manifold $dS_2 \times \mathbb{S}^2$
endowed with the  product metric. This spacetime is a solution of the $\Lambda$-vacuum Einstein equations with cosmological constant
$\Lambda = \varkappa^2$.
It is straightforward to check that
the Killing algebra is six dimensional with a basis consisting 
on  three
linearly independent Killings vectors of $dS_2$ and
three independent  Killing vectors on the sphere. In standard global 
coordinates of $dS_2$ the Nariai metric takes the form
\begin{align*}
g_N= - dt^2 + \cosh^2(\varkappa t) dx^2 + \frac{1}{\varkappa^2} \gamma_{\mathbb{S}^2}
\end{align*}
where $\gamma_{\mathbb{S}^2}$ is the standard unit metric on the sphere.
The most general Killing vector of this metric is given by
\begin{align*}
\zeta = \left ( A \cos ( \varkappa x) + B \sin (\varkappa x) \right ) \partial_t
+ \left[ \beta + ( B \cos (\varkappa x) - A \sin (\varkappa x) ) \tanh (\varkappa t) 
\right ] \partial_x + \widehat{\zeta}
\end{align*}
where $\widehat{\zeta}$ is a Killing vector on
$(\mathbb{S}^2,\gamma_{\mathbb{S}^2})$.
We consider the null hypersurface $\H$
defined as the connected component of
$\tanh (\varkappa t) - \sin(\varkappa x) =0$ containing $t=x=0$. 
%Let $t(x)$ be the function defined by this implicit equation.
Observe that the range of $x$ is given by
\begin{align}
x \in \left ( - \frac{\pi }{2\varkappa}, \frac{\pi }{2\varkappa} \right ).
\label{range}
\end{align} 
By construction $\H$ contains the sphere at $\{ t=x=0\}$.
Topologically $\H \simeq\mathbb{R} \times \mathbb{S}^2$. 
The null generator of $\H$ is
\begin{align*}
k= \partial_x + \frac{1}{\cos (\varkappa x)} \partial_t.
\end{align*}
It is immediate to check that the most general Killing vector
that is proportional to $k$ on $\H$ is given by
\begin{align}
\zeta = \left ( A \cos (\varkappa x) + B \sin (\varkappa x) \right ) \partial_t
+ \left[ A + \left ( B \cos (\varkappa x) - A \sin (\varkappa x) \right )
\tanh (\varkappa t) \right ] \partial_x.
\label{generator}
\end{align}
On the sphere $S_0 := \{ t=x=0 \}$, the Killing vector (\ref{generator}) evaluates to
\begin{align*}
\zeta|_{S_0} = A \left ( \partial_t + \partial_x \right ).
\end{align*}
Thus 
\begin{align*}
\xi := \sin (\varkappa x) \partial_t + \cos(\varkappa x) \tanh(\varkappa t) \partial_x
\end{align*}
is a Killing vector for which $\H \setminus S_0$ is
a non-degenerate Killing horizon with bifurcation surface at $S_0$.
The linearly independent Killing vector
\begin{align*}
\eta := \cos (\varkappa x) \partial_t + \left[ 1 - \sin (\varkappa x)
\tanh(\varkappa t) \right] \partial_x
\end{align*}
vanishes nowhere in the spacetime, in particular on $\H$. 
The corresponding surface
gravity vanishes. This follows immediately from the fact that 
the square norm of $\eta$ can be written as
\begin{align*}
g(\eta,\eta) = \cosh^2(\varkappa t) \left[ \tanh(\varkappa t) - \sin(\varkappa x) \right]^2
\end{align*}
which has a zero of order two at $\H$.
Thus $\H$ is a degenerate Killing horizon of $\eta$, and given that the closure of $\H\setminus S_0$ is $\H$, $\H$ is a MKH of order two ---a double Killing horizon. 

A direct calculation gives $[\xi,\eta] =-\varkappa \eta$ and thus, according to theorem \ref{th:algebra}, the surface gravity of $\H$ is $\k =-\varkappa=-\sqrt{\Lambda}$.

\subsection{Fully degenerate MKHs of any order}\label{Fully}
We want to ascertain if fully degenerate MKHs exist, and which orders are feasible for them. In this section we provide explicit examples for fully degenerate MKHs of any admissible order $m$.

To that end, we use the following construction. In subsection \ref{Minkowski} we found MKHs of maximal order $m=n+1$. The idea is then to try to retain (part or all of) the Abelian subgroup $\Kill_\H^{deg}$ which is generated by $\{\eta_i\}$ with $i=2,\dots, n+1$ in (\ref{eta2})-(\ref{etaA}), but removing the non-degenerate Killing vector (\ref{eta1}) that generates the bifurcate Killing horizon. To accomplish this, we perform a conformal transformation of the flat metric (\ref{gflat}), that is
\be
g =\Omega g^\flat, \label{confg}
\ee
where $\Omega :\mathbb{R}^{n+1} \rightarrow \mathbb{R}$ is a smooth non-vanishing function. To keep $\eta_2$ as a Killing vector of $g$ we require
\be
\pounds_{\eta_2} g = \pounds_{\eta_2} (\Omega g^\flat) = g^\flat \eta_2(\Omega) = g^\flat \left(\partial_t \Omega +\partial_{x^1} \Omega \right)= 0 \hspace{3mm} \Longrightarrow \hspace{2mm} \Omega (t-x^1,x^A).
\label{t-x}
\ee
Similarly, to keep any of the $\eta_{A+1}$%\tim{I guess this should be $\eta_{A+1}$ again, also in the equation below}
as Killing vectors of the metric (\ref{confg}) we demand
$$
\pounds_{\eta_{A+1}} g = \pounds_{\eta_{A+1}} (\Omega g^\flat) = g^\flat \eta_{A+1}(\Omega) = g^\flat \left[x^A\left(\partial_t \Omega +\partial_{x^1} \Omega \right)+(t-x^1) \partial_{x^A}\Omega\right]= 0
$$
and using here (\ref{t-x}) 
\be
\partial_{x^A} \Omega =0. \label{xA}
\ee
Hence, by allowing $\Omega$ in (\ref{t-x}) to be independent of a number $q\leq n-1$ of the variables $\{x^A\}$ we have that the corresponding $q$ vector fields $\eta_{A+1}$ 
%\tim{the same}
are Killing vectors of the new metric $g$. As null hypersurfaces and null vectors are preserved by conformal transformations (\ref{confg}), we know that all these ``surviving'' Killing vectors together with $\eta_2$ are tangent to and null on $\H :=\{t=x^1\}$. On the other hand, the remaining $\eta_1$ in (\ref{eta1}) is not a Killing vector in general, because using (\ref{t-x}) 
$$
\pounds_{\eta_1} g = \pounds_{\eta_1} (\Omega g^\flat) = g^\flat \eta_1(\Omega) = g^\flat \left(x^1\partial_t \Omega +t\partial_{x^1} \Omega \right)= (x^1-t)\,  \partial_t\Omega\,  g^\flat  \neq 0
$$
which is non-vanishing in general ---as long as $\Omega$ has non-trivial dependance on $t-x^1$.

We still need to check that the kept Killing vectors have vanishing surface gravity on $\H$, but this must be the case due to theorem \ref{th:algebra} because they all commute. To check it explicitly though, simply notice that for every $\eta_i$
$$
g(\eta_i,\eta_i) =\Omega g^\flat (\eta_i,\eta_i)  \hspace{5mm} \Longrightarrow \hspace{3mm} \mbox{grad} (g(\eta_i,\eta_i)) =\mbox{grad}\Omega \, g^\flat (\eta_i,\eta_i) +\Omega\,  \mbox{grad}(g^\flat (\eta_i,\eta_i)) \stackrel{t=x^1}{=} 0.
$$

The case of maximal order, that is with $m=n$ (so $q=n-1$), has a conformal factor $\Omega (t-x^1)$, and these metrics
\be
g =\Omega (t-x^1) g^\flat, \label{PW}
\ee
describe conformally flat plane waves, known to be solutions of the Einstein-Maxwell equations \cite{Exact} for a null electromagnetic field $F=(dt-dx^1)\wedge v_A  dx^A$ where $v_A$ are functions of $t-x^1$ ---and more generally these are solutions of the Einstein-$p$-form equations for a null $p$-form, arising in higher dimensional theories such as supergravity. 
By using null coordinates
$$
U= t-x^1 , \hspace{1cm} V=t+x^1
$$
the metric can be written in the forms
\be
g=\Omega(U) \left(-dUdV  +\sum_{A=2}^n (dx^A)^2 \right)= -du dV +\Omega(u) \sum_{A=2}^n (dx^A)^2 \label{PW1}
\ee
where $\Omega(U) dU :=d u$. The last expression is the canonical Einstein-Rosen form of the (conformally flat) plane wave. Every null hypersurface $u =$const.\ is a fully degenerate MKH of maximal order $m=n$ in these spacetimes.

As is well known, plane waves such as (\ref{PW1}) can be cast (and actually extended through removable singularities arising at the zeros of $\Omega(u)$) in Kerr-Schild form, where the spacetime is geodesically complete. The extension is given by the new set of coordinates $\{u,v,z^A\}$ defined by (an overdot means derivative with respect to $u$)
$$
V= -2v -\frac{\dot\Omega}{2\Omega}\sum_{A=2}^n (z^A)^2, \hspace{1cm} x^A =\frac{1}{\Omega^{1/2}} z^A, 
$$
so that (\ref{PW1}) becomes
\be
g = 2dudv +\Psi(u) \delta_{AB}z^A z^B du^2  +\sum_{A=2}^n (dz^A)^2 
\ee
with
$$
\Psi (u) := \frac{\ddot\Omega}{2\Omega} - \frac{\dot\Omega^2}{4\Omega^2} .
$$

\subsection{Ricci-flat metrics with fully degenerate MKHs}
Now that we know that fully degenerate MKHs exist and can have any order, we wish to present an example of a spacetime which contains a fully degenerate MKH and solves the vacuum Einstein field equations, that is, its Ricci tensor vanishes. The previous subsections showed us that perhaps plane waves are good candidates for this purpose. Therefore, let us consider the most general {\em vacuum} (i.e. Ricci flat)  plane wave, given by
\be
g= 2dudv + M_{AB}(u) z^A z^Bdu^2+\sum_{A=2}^n (dz^A)^2 , \hspace{1cm} \delta^{AB} M_{AB} =0
\label{PW3}
\ee
where $M_{AB}(u)$ is a trace-free symmetric matrix of functions of $u$. 
%%
%\begin{equation*}
%g = H(u,y,z) \mathrm{d}u^2 + 2 \mathrm{d}u \mathrm{d}v +  \mathrm{d}y^2 +  \mathrm{d}z^2
%\,,
%\end{equation*}
%%
%and  provides a solution to  Einstein's vacuum field equations, vacuum plane waves, supposing that
%%
%\begin{equation*}
%H(u,y,z)= (y^2-z^2)f^{(1)}(u) - 2yz f^{(2)}(u)
%\,.
%\end{equation*}
%
To exclude the Minkowski case we assume that rank$(M_{AB})\geq 1$.
%$f^{(1)}$ and $f^{(2)}$ are not both identically zero.

The most general Killing vector field $\zeta$ for (\ref{PW3}) reads
\be
\zeta= (a_0+a_1 u) \partial_u +(b-a_1 v -\dot{c}_A z^A)\partial_v +(c_A(u)+\epsilon_{AB} z^B)\partial_{z^A}
\ee
where $a_0,a_1,b$ and $\epsilon_{AB}=-\epsilon_{BA}$ are real constants, and $c_A(u)$ are functions satisfying %\marc{correction, please check \\ -- \\ tim: checked}
\bea
(a_0 +a_1 u) \dot{M}_{AB} +2a_1 M_{AB} &=& 
\epsilon_{AC}M^{C}{}_B + \epsilon_{BC} M^{C}{}_A
, \label{cond_a1}\\
\ddot{c}_A &=& M_{AB}c^B .\label{ODE_b}
\eea
where $A,B$ indices are raised with $\delta^{AB}$.
%%
%\begin{eqnarray*}
%\zeta = \begin{pmatrix}
%a_0+ a_1u
%\\
% p_0-a_1v  - b'(u) y  - c'(u)z
%\\
% b(u) + d_0 z 
%\\
% c(u) - d_0 y 
%\end{pmatrix}
%\,, \quad a_0,a_1,p_0, d_0\in\mathbb{R}
%\,,
%\end{eqnarray*}
%%
%with 
%%
%\begin{eqnarray}
%b''(u) &=& b(u)f^{(1)}(u)- c(u)f^{(2)}(u)
%\,,
%\label{ODE_b}
%\\
%c''(u) &=&b(u)  f^{(2)}(u)-c(u)f^{(1)}(u)
%\,,
%\label{ODE_c}
%\\
%(a_0+ a_1u) {f^{(1)}}'(u) +  2  a_1 f^{(1)}(u) &=& -2d_0f^{(2)}(u)
%\,,
%\label{cond_a1}
%\\
%(a_0+ a_1u){f^{(2)}}'(u)+2a_1 f^{(2)}(u)   &=& 2 d_0f^{(1)}(u)
%\,.
%\label{cond_a2}
%\end{eqnarray}
%%
Hence, the spacetime has at least $2(n-1)+1=2n-1$ Killing vectors which are determined by the parameters $c_A^0:=c_A(0)$ and $c_A^1:=\dot {c}(0)$, which
are the initial data for the 2nd-order ODEs (\ref{ODE_b}), plus $b$.
There might be additional Killing vectors
depending on whether or not $M_{AB}(u)$ is such that (\ref{cond_a1}) admits a non-trivial solution for the constants $(a_0,a_1,\epsilon_{AB})$.

The candidates to a MKH are the hypersurfaces $u=$const. Without loss of generality, let us consider the null hypersurface $\mathcal{H}:=\{u=0\}$, and we are interested in those Killing vectors for which this is a horizon.
This will be the case if and only if $a_0=0$, $\epsilon_{AB}=0$ and $c_A^0=0$.
%$a_0=b_0=c_0=d_0=0$. 
In that case, $a_1=0$ is also required as otherwise $M_{AB} \propto u^{-2}$, which would be singular at $\mathcal{H}$.
Thus the most general Killing vector in $\Kill_\H$ is given by
$$
\eta = (b-\dot{c}_A z^A)\partial_v +c_A(u)\partial_{z^A}
$$
where all the $c_A$ vanish at $u=0$. In particular
$$
\eta |_\H = (b-c^1_A z^A)\partial_v .
$$
Notice that $g(\eta,\eta)=c_A(u) c^A(u)$ whose gradient vanishes at $u=0$, and thus all the surface gravities are zero, so that $\Kill_\H =\Kill_\H^{deg}$ 
%$\eta$  which is compatible with these requirements is thus the 3-parameter family
%\begin{equation*}
%\eta = \begin{pmatrix} 0 \\ -b' (u) y - c' (u) z +p_0 \\ b(u) \\ c(u)\end{pmatrix}
%\,, \quad \text{with}\quad   \eta |_{\mathcal{H}} = \begin{pmatrix} 0 \\ -b_1 y - c_1 z +p_0 \\  0 \\  0 \end{pmatrix}
%\,.
%\end{equation*}
%
%We have $\eta^{\nu}\nabla_{\nu}\eta^{u}|_{\mathcal{H}}  =0$, so $\mathcal{H}=\{u=0\}$ 
and $\H$ is a fully degenerate MKH of order $n$.

\subsection{Near Horizon Geometry: double (or higher) Killing horizons}\label{NH}

Observe that in the previous example with conformally flat metric (\ref{confg}) we could have also kept the non-degenerate Killing $\eta_1$ in (\ref{eta1}) had we also requested that $\Omega_{,t}=0$, and in that way we would obtain MKHs of any order and non-fully degenerate. This is actually a completely general property of fully degenerate Killing horizons of any order $m$, in the sense that they can be promoted to non-fully degenerate MKHs of order at least $m+1$. This general construction will be discussed in the next subsection, as it involves the so-called near horizon geometry which we analyze next.

The metric of any {\it Near Horizon geometry} \cite{KL}, which can be thought of as the ``focused'' local geometry near any degenerate Killing horizon, actually possesses a non-fully degenerate MKH. This can be seen from the explicit expression of the metric in local coordinates $\{u,v,x^A\}$
\be
g_{NH} = 2dv\left(du + 2 u\,  s_A dx^A +\frac{1}{2} u^2 H dv\right) +\gamma_{AB} dx^A dx^B \label{NHG}
\ee
where $\gamma_{AB}$ is the metric on any cut $S\subset \H$, $s_A$ is a one-form on $S$, and $H$ a smooth function on $S$, while the degenerate Killing horizon is given by $\H=\H_{\eta} := \{u=0\}$, where the Killing vector is $\eta=\partial_v$. Observe that $g(\eta,\eta)=u^2 H$ so that $\eta$ is null on $\H_{\eta}$ and obviously $\kappa_{\eta}=0$. 

As noted in \cite{KL}, see also \cite{PLJ,LSW}, the metric (\ref{NHG}) always has another Killing vector given by
$$
\xi = v\partial_v -u\partial_u.
$$
Obviously this Killing vector is null on $\H_{\eta}$ 
and tangent to it, except at $S:=\{u=v=0\}$ where it vanishes. Thus, $\H_{\xi} = \H_{\eta}\setminus \{v=0\}$ is also a Killing horizon for $\xi$, with two connected components and $\overline{\H}_{\xi}=\H_{\eta}$, hence $\{u=0\}$ is (at least) a double Killing horizon. A direct calculation provides 
$$
[\xi,\eta] = -\eta
$$
so that, from theorem \ref{th:algebra} follows that $\{u=0\}$ is non-fully degenerate and that $\k =1$, while $\tau =\ln |v|$. The Killing $\xi$ generates a bifurcate Killing horizon based on $S$ with branches given by $\{u=0\}$ and $\{v=0\}$. Actually, a bifurcate Killing horizon is defined by any cut $\{u=0,v=v_0\}$ on $\H_{\eta}$, with bifurcation Killing vector $\xi -v_0\eta $.

\subsection{From fully to non-fully degenerate MKHs}
The results of the previous subsection provide a method to generate a non-fully degenerate MKH starting from a fully degenerate one. Moreover, combining this method with the results of subsection \ref{Fully} we can construct non-fully degenerate MKHs of any order explicitly. 

The idea consists in taking a fully degenerate MKH of order $m$, and then computing its near horizon geometry (\ref{NHG}). This always provides a non-fully degenerate MKH as seen in subsection \ref{NH}, but to ensure that the construction works we need to check that none of the multiple Killing vectors of the original MKH is lost in the process. And this follows from a classical and very elegant argument by Geroch \cite{G} concerning hereditary properties when taking limits of one-parameter families of spacetimes. Geroch proved that, given a family $(M_\lambda, g_\lambda)_{\lambda>0}$ of spacetimes depending on a continuous parameter $\lambda$ and all of them having $q$ linearly independent Killing vectors, then the limit spacetime defined by taking the limit when $\lambda \rightarrow 0$ (when this limit exists) also has, at least, $q$ linearly independent Killing vectors. This result applies to our construction because the near horizon geometry (\ref{NHG}) is actually defined as follows: nearby a {\em degenerate} Killing horizon, there exist local Gaussian null coordinates such that the metric takes the form
$$
g = 2dv\left(du + 2 u\,  \hat s_A dx^A +\frac{1}{2} u^2 \hat H dv\right) +\hat \gamma_{AB} dx^A dx^B 
$$
where now $\hat H$, $\hat s_A$, and $\hat\gamma_{AB}$ may all depend on $u$ too: they are functions depending on $u$ and $x^A$. Defining the one-parameter family of metrics $\{g_\lambda\}_{\lambda>0}$ by replacing $v\rightarrow v/\lambda$ and $u\rightarrow u\lambda$ and taking the limit $\lambda \rightarrow 0$ leads to the metric (\ref{NHG}) where $H=\hat H|_{u=0}$, $s_A=\hat s_A|_{u=0}$ and $\gamma_{AB}=\hat\gamma_{AB}|_{u=0}$. 

Of course, it could still happen that one of the Killing vectors was lost in the limiting process, and ``replaced" by the new one that the near horizon limit always adds. But this is not possible in the case where the original group of motions is Abelian, as the only possibility is a contraction of the Lie algebra in the sense of \cite{IW,Saletan}, see \cite{BFP}: a higher (or equal) number of structure constants vanish after the limit. Thus, due to theorem \ref{th:algebra}, if we start with a fully degenerate MKH of order $m$, its near horizon geometry necessarily will have a MKH of order (at least) $m+1$.
%\jose{Should we make this into a theorem? I think so...} \marc{done}
This line of reasoning also proves that, for any non-fully degenerate MKH of order $m\geq 3$, its near horizon geometry has a MKH of order at least $m$. 
Summarizing, we have established the following result.
\begin{theorem}
\label{NearHor}
Let $\H$ be a multiple Killing horizon of order $m$ and
$(M_{\mbox{\tiny NHG}}, g_{\mbox{\tiny{NHG}}})$ 
be the near-horizon 
geometry of a degenerate Killing vector $\eta$ of $\H$. Then
\begin{itemize}
\item[(i)] If $\H$ is fully degenerate, 
$(M_{\mbox{\tiny NHG}}, g_{\mbox{\tiny{NHG}}})$ admits a multiple Killing horizon
of order at least $m+1$.
\item[(ii)] If $\H$ is non-fully degenerate and $m \geq 3$, then
$(M_{\mbox{\tiny NHG}}, g_{\mbox{\tiny{NHG}}})$  has a multiple Killing horizon
of order at least $m$.
\end{itemize}
\end{theorem}

\begin{remark}
Item (ii) implies, in particular, that if $\H$ is of maximal order, then
any near horizon geometry that one may construct from it must
also be of maximal order.
\end{remark}

\begin{remark}
%\tim{needs rewording, depending on whether a result can be announced or not}
A natural question is whether the NHG spacetime 
$(M_{\mbox{\tiny NHG}}, g_{\mbox{\tiny{NHG}}})$ arising from
a multiple Killing horizon is independent 
of the choice of degenerate Killing vector $\eta$.
% Though this is not obvious at first sight, it seems that this might be the case in general. 
%It turns out that in general it is not and that a whole
%class of non-isometric NHG can be constructed  
%from a multiple Killing horizon (either fully degenerate or
%non-fully degenerate of order $m \geq 3$). One can also analyze 
%the necessary and sufficient conditions under which the NHG limit
%spacetime is independent of the choice of $\eta$. 
%These results will be presented elsewhere.
This question will be addressed, with a thorough analysis, elsewhere.
\end{remark}

To illustrate the points in Theorem \ref{NearHor}, let us carry over the construction explicitly for the fully degenerate MKHs of subsection \ref{Fully}. Starting with the metric (\ref{confg}), and assuming that $\Omega (t-x^1,x^{A'})$ is independent of a number $q\geq 1$ of the coordinates $\{x^A\}$ so that $\H:=\{t=x^1\}$ is a fully degenerate MKH of order $q+1$, we need to construct its near horizon geometry (\ref{NHG}). To be explicit, we split the set of coordinates $\{x^A\}$ into two subsets $\{x^A\}:=\{x^{A'}, y^\Upsilon\}$ and use the notation 
\begin{align*}
\{x^{A'}\} & =\{x^A\}_{A=2,\dots ,n-q} \quad \quad  (\Longrightarrow A',B',\dots \in\{2,\dots ,n-q\}) ,\\
\{y^\Upsilon\} & = \{x^A\}_{A=n+1-q, \dots , n} \quad \quad  (\Longrightarrow \Upsilon \in\{n+1-q,\dots ,n\}).
\end{align*}
For the construction,
%\jose{Is it clear that MKHs have the same near horizon geometry for all the degenerate Killings? Should we prove it?} 
we choose $\eta =\eta_2$ as the degenerate Killing vector (because $\H=\overline\H =\H_{\eta_2}$) and, instead of looking for Gaussian null coordinates around the MKH, we can use the fact that $F,s_A$ and $\gamma_{AB}$ have a clear geometric interpretation as follows: let $S_0\subset \H$ be the co-dimension two submanifold defined by $\{u=0,v=v_0\}$ in the metric (\ref{NHG}) . Then \cite{KL}
\begin{itemize}
\item $\gamma_{AB}$ is the inherited metric on $S_0$
\item $s_A$ is the torsion one-form on $S_0$, defined by $\bm{s}(V) := \bm{\ell} (\nabla_V \eta_2)$ for any $V\in \mathfrak{X}(S_0)$, where $\ell$ is the unique null vector field orthogonal to $S_0$ satisfying $g(\ell,\eta_2)=-1$.
\item $H=2 \gamma(s,s) - \mbox{div} \, s + \frac{1}{2} R|_{S_0}- 
\frac{1}{2}\gamma^{AB}R_{AB} |_{S_0}$
\end{itemize}
where div is the divergence on $S_0$, $R$ is the scalar curvature and $R_{AB}$ are the $AB$-components of the Ricci tensor of $(M,g)$, both evaluated
at $S_0$.

Set by definition $\Omega_0 (x^{A'}) := \Omega |_\H = \Omega|_{t=x^1}$. Then the metric on $S_0$ reads simply
\be
\gamma =\Omega_0 \gamma^\flat \label{confgamma}
\ee
where $\gamma^\flat$ is the flat $(n-1)$-dimensional Euclidean metric. 
Noting that $\bm{\ell}=-\frac{1}{2}(dt+dx^1)$, a straightforward calculation shows that
\be
s=- \frac{1}{2}d\ln\Omega_0 %, \hspace{1cm} (s_A =\partial_{x^A}(\ln \Omega_0) ) 
\label{s=d}
\ee
and therefore only the components $s_{A'}$ are non-identically vanishing. Finally, to compute $H$ we use the fact that $\eta_2$ is null everywhere for the metric (\ref{confg}), and this must be kept for its near horizon geometry, so that $\eta =\eta_2$ is null everywhere. But $g_{NH}(\eta,\eta)=u^2H$, and thus $H$ necessarily vanishes.
%we could use following expression for the scalar curvature of (\ref{confg})
%$$
%R =-\frac{n}{\Omega} \left({}^{\flat}\BoxÂ \ln \Omega +\frac{n-1}{4} g^\flat (\mbox{grad}\ln \Omega, \mbox{grad}\ln \Omega) \right)
%$$
%where ${}^{\flat}\Box$ is the D'Alembert operator on the flat spacetime. Using here the explicit dependence of $\Omega =\Omega(t-x^1,x^{A'})$ this implies readily
%$$
%R|_{S_0} = -n\left(\frac{1}{\Omega_0} {}^{\flat}\Delta \ln\Omega_0 +\frac{n-1}{4} \gamma (s,s)\right)
%$$
%where ${}^{\flat}\Delta$ is the flat Laplace operator in $(n-2)$ dimensions. A similar calculation produces
%$$
%\gamma^{AB}R_{AB} =\frac{n-1}{n} R|_{S_0}
%$$
%while on using (\ref{s=d}) 
%$$
%\mbox{div}\, s = -\frac{1}{\Omega_0}{}^{\flat}\Delta \ln\Omega_0 -\frac{n-3}{2} \gamma (s,s).
%$$
%Collecting everything we obtain\jose{Should be checked... It must be zero}
%\be
%F= -\frac{1}{\Omega_0} \left(\frac{n-1}{8} \delta^{AB}\partial_{x^A}\ln\Omega_0\,  \partial_{x^B}\ln\Omega_0 \right) \label{F}
%\ee

Hence, the near-horizon limit of (\ref{confg}) with $\Omega(t-x^1,x^{A'})$ leads us to the metric
%\tim{corrected (sum goes up to $n$)}
$$
g_{NH} = 2dv\left(du - u\,  d\ln\Omega_0 \right) +\Omega_{0}\sum_{A=2}^{n} (dx^A)^2 ,
$$
with $\Omega_0(x^{A'})$ any arbitrary positive function independent of the $q$ coordinates $\{y^\Upsilon\}$ among the $\{x^A\}$. This metric has a non-fully degenerate MKH $\H:=\{u=0\}$ of order $q+2$.

To crosscheck that the construction works fine, we can exhibit the Killing vectors generating $\Kill_\H$, which are given by %\marc{corrected}
$$
\zeta = (av +c_\Upsilon y^\Upsilon +b)\partial_v -a u \partial_u - 
\frac{1}{\Omega_0} u c^\Upsilon \partial_{y^\Upsilon}
$$
where $a$, $b$ and $c_\Upsilon =c^\Upsilon$ are ($q+2$) arbitrary constants. At $u=0$ all of them are proportional to $\eta_2=\partial_v$, and setting $a=0$ we get $\Kill_\H^{deg}$. 

The case with $\Omega_0=$const.\ is flat spacetime, which thus arises as the near-horizon geometry of the maximal fully degenerate MKH in the conformally flat plane waves (\ref{PW1}). 

\section{The master equation for MKHs}\label{sec:master}
%\marc{Several changes involving $\H \rightarrow \H_{\xi}$ and similar}
In this section we look for an equation that the proportionality function between different Killing vectors of a given MKH must satisfy.

Let $\H$ be a multiple Killing horizon, and using the notation of Remark \ref{notation} let $\eta\in \Kill^{deg}_\H$ and $\xi\in \Kill_\H$, so that $\kappa_\eta =0$ and $\k$ can be zero or not. Given that $\xif$ is normal and non-zero 
on its corresponding Killing horizon $\H_{\xi} \subset \H$,
we know that on $\H_{\xi}$ there exists a one-form $\bm{\Phi}\in T_{\H_{\xi}}^* M$ such that (for a proof, see \cite{FN})
\be
d \xif \eqH 2\bm{\Phi}\wedge \xif \label{dxi}
\ee
or equivalently
\be
\nabla_\mu \xi_\nu \eqH \Phi_\mu \xi_\nu -\Phi_\nu \xi_\mu \, .\label{dxi1}
\ee
The one-form $\bm{\Phi}$ cannot stay bounded at the zeros of $\xi $ in $\H$ (if any)
because $d \xi$ and $\xi$ cannot vanish simultaneously at any point.
$\bm{\Phi}$ is not univocally defined, as there is the gauge freedom
\be
\bm{\Phi} \rightarrow \bm{\Phi} + {\cal B}  \xif , \label{gauge}
\ee 
for an arbitrary smooth function ${\cal B}:\H_{\xi} \rightarrow \mathbb{R}$. Contracting (\ref{dxi}) with any $X\in \mathfrak{X}(\H_{\xi})$ (and thus 
%\marc{small change}
fulfilling the condition $\xif (X)=0$) one obtains
\be
\nabla_X \xi \eqHxi \bm{\Phi} (X) \xi , \qquad \forall X\in \mathfrak{X}(\H_{\xi}) \label{nablaVxi}
\ee
and, in particular, for $X=\xi$
\be
\nabla_\xi \xi \eqH \bm{\Phi} (\xi) \xi, \quad \Longrightarrow \quad
 \bm{\Phi} (\xi) = \k  .\label{nablaxixi}
\ee
Incidentally, this provides a proof of expression (\ref{kappasquare}) by just squaring (\ref{dxi1}).
%The last equality tells us that there are two slightly different situations, depending on whether $\H$ is fully degenerate or not:
%\begin{enumerate}
%\item Non-fully degenerate $\H$ ($\k \neq 0$). In this case, any of the vector fields $\Phi$ is transversal everywhere to $\H$ and thus it could be used as a {\em rigging} \cite{MS} for the null hypersurface $\H$. In particular, using the freedom (\ref{gauge}) one can always choose $\Phi$ to be null, or more in general to have a constant norm $g(\Phi,\Phi)$. With any of these choices the freedom (\ref{gauge}) is completely fixed. Observe that, in this case, $\xi =0$ on a bifurcation $S\in \H$, and given that then $d\xif$ cannot vanish at any point of $S$ the one-form $\bm{\Phi}$ is not defined on $S$: it diverges there in such a way that the exterior product (\ref{dxi}) remains finite and non-zero all over $S$.
%\item Fully degenerate $\H$ ($\k=0$). In this situation $\Phi$ is actually tangent to $\H$. The case with vanishing $\Phi$ ---or more generally, proportional to $\xi$---is included here. Whenever $\Phi$ is not proportional to $\xi$ the freedom (\ref{gauge}) can be used to make $\Phi$ tangent to any desired cut in $\H$. The same reasoning as before proves that $\Phi$ must diverge at any possible fixed points of $\xi$ in $\H$.
%\end{enumerate}
Comparing (\ref{varphi}) with (\ref{nablaVxi}) we observe that 
$$
\bm{\Phi}(X) =\bm{\varphi} (X), \hspace{1cm} \forall X\in \mathfrak{X}(\H_{\xi}).
$$
Furthermore, from (\ref{nablaVxi}) and (\ref{nablaxixi})
$$
\bm{\Phi} (\nabla_X \xi) =\k \bm{\Phi} (X) =\k \bm{\varphi}(X) , \hspace{1cm} \forall X\in \mathfrak{X}(\H_{\xi})
$$
which, upon using the constancy of $\bm{\Phi}(\xi)$ on $\H_{\xi}$,
 can be rewritten as
\begin{align}
\xif (\nabla_X \Phi +\bm{\Phi} (X) \Phi) =0
\label{tangency}
\end{align}
so that the vector fields $\nabla_X \Phi +\bm{\Phi} (X) \Phi$ are tangent to
$\H_{\xi}$ for arbitrary $X\in \mathfrak{X}(\H_{\xi})$. %\marc{addition}
Another consequence of
(\ref{nablaVxi}) is the well-known fact that $\H_{\xi}$ is totally geodesic, i.e. that
given any pair of (spacetime) vector fields $X,Y$ tangent
to $\H_{\xi}$, the vector $\nabla_X Y$ is also tangent to 
$\H_{\xi}$. This means, in particular, that $\nabla_X Y$ makes
sense as a map
\begin{align}
\nabla: \X (\H_{\xi}) \times \X(\H_{\xi}) \longrightarrow \X(\H_{\xi}).
\label{nabla}
\end{align}
%This will be used later.
Similarly, for $\eta\in \Kill_\H^{deg}$ we have, on its corresponding Killing
horizon,
\be
d \bm{\eta}\eqHeta 2 \bm{w}\wedge \bm{\eta}, \label{detai}
\ee
for a one-form $\bm{w}\in T^*_\H M$. As before, this one-form diverges at 
fixed points of $\eta$ on $\overline\H$ (in particular, at
the points $\{ f =0\} \subset \H_{\xi}$), and is defined up to the addition of an arbitrary multiple of $\xif$
\be
\bm{w}\rightarrow \bm{w}+ G \xif, \hspace{1cm} G\in C^\infty(\H_{\eta}) .\label{gauge1}
\ee
A similar calculation as above leads to
\be
\nabla_X \eta\eqHeta \bm{w} (X) \eta , \hspace{1cm} \forall 
X\in \mathfrak{X}(\H_{\eta}) \label{nablaVeta}
\ee
and in particular, using $\kappa_\eta=0$,
$$
\bm{\eta} (w) =0 \quad \Longrightarrow \quad \xif (w)  =0 .
$$
Hence, $w$ is tangent to $\H_{\eta}$ everywhere. 
%Using the freedom (\ref{gauge1}), in the case with $\k\neq 0$ one can further achieve 
%\be
%\bm{\ell} (w_i) =0\label{ellwi}
%\ee
%if desired. 
Using (\ref{etaxi1}) together with (\ref{nablaVxi}) and
 (\ref{nablaVeta}) the following equation follows
 on $\Hhat := \H_{\xi} \cap \H_{\eta}$  
%\marc{Changed: at this point one does not know whether $\bm{\Phi} -\bm{w}$ has any continuous extension, and $f$ typically diverges at $\H \setminus \H_{\xi}$}
%on each connected component of $\H_\xi$, and by continuity everywhere on $\H$
$$
\nabla_X (f e^{-\k \tau}) +f e^{-\k \tau} \left(\bm{\Phi} (X) - \bm{w} (X)\right) 
\eqHhat 0, \hspace{1cm} \forall X\in \mathfrak{X}(\Hhat)
$$
or equivalently
$$
\nabla_X f \eqHhat f \left( \bm{w} (X) +\k X(\tau)-\bm{\varphi}(X) \right)
$$
which provides no new information for $X=\xi$, 
%while for any $V\in \mathfrak{X}(\H)$ with $V(u)=0$ gives
%$$
%\nabla_V f = f(\bm{w}-\bm{\Phi}) (V)
%$$
%and fully determines the one-form $\bm{w}$ in terms of $f$:
%$$
%\bm{w}(V) = V(\ln f) +\bm{\varphi}(V).
%$$
and it fully determines the pullback of the one-form $\bm{w}$ to $\Hhat$: let $\iota :\Hhat\rightarrow M$ be the inclusion of $\Hhat$ into the manifold $M$, and let $\iota^\star$ be its pullback, then the previous expression can be rewritten as
\be
 \iota^\star \bm{w}=d \ln f -\k d\tau +\bm{\varphi}  \label{wiexact}
\ee 
where the exterior derivative should be understood as the one in $\Hhat$ as a manifold. Observe that this relation is unaffected by the gauge (\ref{gauge1}), as $\iota^\star \xif =0$, and that $\iota^\star \bm{\Phi}=\bm{\varphi}$. Expression (\ref{wiexact}) recovers the previous result that $\bm{w}$ is ill-defined at the fixed points of $\eta$ in $\H_{\xi}$, where $f$ has zeros. 

Contracting (\ref{D2xi})  with $X\in \mathfrak{X}(\H_{\xi})$ while using (\ref{dxi})
and (\ref{nablaVxi}) one gets on $\H_{\xi}$
$$
\left( \nabla_X \bm{\Phi} +\bm{\Phi}(X) \bm{\Phi}\right)\wedge \xif \eqH X^\lambda \xi_\rho \bm{\Omega}^\rho{}_{\lambda}, \hspace{1cm} \forall X\in \mathfrak{X}
(\H_{\xi})
$$
where $\bm{\Omega}^\rho{}_{\lambda}:= \frac{1}{2} R^{\rho}{}_{\lambda\mu\nu}
dx^{\mu} \wedge dx^{\nu}$ are the 2-forms of curvature. With indices
\be
(X^\sigma\nabla_\sigma \Phi_\mu +\Phi_\sigma X^\sigma \Phi_\mu)\xi_\nu -
(X^\sigma\nabla_\sigma \Phi_\nu +\Phi_\sigma X^\sigma \Phi_\nu)\xi_\mu 
\eqHxi X^\lambda \xi_\rho R^\rho{}_{\lambda\mu\nu} .\label{riemann1}
\ee 
This implies, on the one hand (by (\ref{tangency}))
\be
X^\lambda \xi_\rho \xi^\mu R^\rho{}_{\lambda\mu\nu}\eqHxi 0, \hspace{1cm} \forall X\in \mathfrak{X}(\H_{\xi})
\label{Xxi2}
\ee
which is nothing else that (\ref{Xxi}) ---as $\k$ is constant---, and on the other hand, for $X=\xi$, the existence of a function ${\cal G}:\H_{\xi} \rightarrow \mathbb{R}$ such that
$$
\nabla_\xi \Phi+\k \Phi = {\cal G} \xi, 
$$
which implies $2{\cal G}\k =2\k (\Phi_\rho \Phi^\rho) +\nabla_\xi (\Phi_\rho \Phi^\rho)$. 
An analogous calculation starting from 
$$
\nabla_\lambda \nabla_\mu \eta_\nu = \eta_\rho R^\rho{}_{\lambda\mu\nu}
$$
leads $\forall X\in \mathfrak{X}(\H_{\eta})$ to
$$
(X^\sigma\nabla_\sigma w_\mu +w_\sigma X^\sigma w_\mu)\eta_\nu -
(X^\sigma\nabla_\sigma w_\nu +w_\sigma X^\sigma w_\nu)\eta_\mu 
\eqHeta X^\lambda \eta_\rho R^\rho{}_{\lambda\mu\nu} .
$$
Introducing here (\ref{etaxi1}) this becomes
\be
(X^\sigma\nabla_\sigma w_\mu +w_\sigma X^\sigma w_\mu)\xi_\nu -
(X^\sigma\nabla_\sigma w_\nu +w_\sigma X^\sigma w_\nu)\xi_\mu 
\eqHeta X^\lambda \xi_\rho R^\rho{}_{\lambda\mu\nu}. \label{riemann2}
\ee
which is an alternative expression for the righthand side of (\ref{riemann1}).
For $X=\xi$ this gives
$$
\nabla_\xi w=\hat G  \xi  \hspace{1cm} \Longrightarrow \hspace{1cm}
\nabla_\xi (w_\mu w^\mu) =0.
$$

Combining the two expressions (\ref{riemann1}) and (\ref{riemann2}) we get,
on $\Hhat = \H_{\xi} \cap \H_{\eta}$
\be
Y^\mu\left( X^\sigma\nabla_\sigma \Phi_\mu +\Phi_\sigma X^\sigma \Phi_\mu-
X^\sigma\nabla_\sigma w_\mu -w_\sigma X^\sigma w_\mu\right) =0, \hspace{1cm} \forall X,Y\in \mathfrak{X}(\Hhat) .\label{master}
\ee
This expression together with (\ref{wiexact}) provides a second order PDE for the function $f$ which is the basic fundamental equation of MKHs. We call it the {\em master equation}.

\subsection{The master equation as a PDE on any cut of $\H$}
The contraction of either (\ref{riemann1}) or (\ref{riemann2}) with $\xi^\nu$ gives no information due to (\ref{Xxi2}). Similarly, contraction with two vector fields tangent to $\H$ gives known information, namely that certain
components of the Riemann tensor vanish on $\H$.
 Thus, the relevant information contained in either (\ref{riemann1}) or (\ref{riemann2}) is given by contraction with a vector field transversal to $\H$ everywhere, and a vector field tangent to $\H$ but different from $\xi$. Concerning the master equation (\ref{master}), both vectors $X$ and $Y$ must be different from $\xi$ for it to yield a non-trivial equation.
To extract this information we work on $\widehat \H$, and select a scalar function
$\tau: \Hhat \longrightarrow \mathbb{R}$ as in Section \ref{basics}, i.e. 
satisfying $\xi(\tau)=1$. %\marc{redistribution and some changes}
The level sets of this function defines a 
foliation $\{S_\tau\}$ of $\widehat \H$ by spacelike co-dimension two surfaces.
By restricting $\Hhat$ if necessary we may assume that it admits
a cross section, i.e. a spacelike codimension-two surface
crossed  precisely once by each inextendable null generator. Under this assumption, 
the freedom  (\ref{ufreedom}) implies that one of the leaves of the foliation
can be selected arbitrarily, and then the whole foliation is uniquely fixed.
Everything that follows is valid for 
any such choice of $\tau$.
%This is equivalent to selecting a foliation $\{S_\tau\}$ of $\widehat \H$ with each $S_\tau$ being defined by a constant value of $\tau$. What follows is valid for any other choice of $\tau$ according to (\ref{ufreedom}), and therefore, for any choice of foliation by spacelike co-dimension two surfaces of $\widehat\H$

%and define $\bm{\ell} \in \X^{\star} (\Hhat)$ 
%\be
%\bm{\ell} =-d\tau . \label{ell}
%\ee
%Observe that $\bm{\ell} (\xi)=-1$ and that $\pounds_\xi \bm{\ell}=0$ (here $\pounds$ is the Lie derivative on the manifold $\Hhat$), and 

Define the set of vector fields associated to the foliation $\{ S_{\tau}\}$
%\marc{slight change of notation to avoid too much abuse}
$$
\mathfrak{X} (\{ S_\tau\} ) :=\{ V\in \mathfrak{X}(\widehat\H), \hspace{3mm} [\xi,V]=0, \hspace{3mm} V(\tau)=0\}.
$$
%\marc{Several additions in about one page}
Note that any vector field $X_0$ in a given leaf $S_{\tau_0}$, $X_0
\in \X(S_{\tau_0})$  gives
rise to an element $X \in \X(\{ S_{\tau} \}$ by simply solving $[\xi,X] =0$
with initial data $X_0$. Conversely, any $X \in \X(\{ S_{\tau} \})$
defines a vector field $X_0$ tangent to $S_{\tau_0}$ by simply $X_0 := X |_{S_{\tau_0}}$. It is immediate to check that this is an isomorphism (see
\cite{MarsSoria} for further details). An easy consequence of this
isomorphism  is that  $\bm{\omega} \in \Lambda (\Hhat)$ vanishes
if and only if it vanishes on $\xi$ and all $V \in \X(\{S_{\tau}\})$. 
We make the statement explicit for later use
\begin{align}
\label{Property}
\mbox{For } \bm{\omega} \in \Lambda(\Hhat): 
\quad \quad \quad \quad 
\left . 
\begin{array}{ll}
\bm{\omega}(V)&=0 \quad \quad  \forall V \in \X (\{ S_{\tau}\}) \\
\bm{\omega}(\xi) &=0 
\end{array}
\right \}
\quad \quad \Longleftrightarrow \quad \quad \bm{\omega} =0.
\end{align}
The following fact will also be needed. Let $\ell \in T_{\Hhat} M$ be
a vector field in $M$ along $\Hhat$, uniquely defined by the conditions of 
being  null, orthogonal to $S_{\tau}$ $\forall \tau$
and satisfying $g(\ell,\xi)=-1$. It follows immediately that
\begin{align*}
\iota_{\star} (\bm{\ell} ) = - d\tau:= \widehat{\bm{\ell}}.
\end{align*}

The Lie derivative along $\xi$ commutes with the
spacetime covariant derivative $\nabla$, and this property descends to
$\nabla_V W$ when this operation is viewed as in (\ref{nabla}). Hence 
$[ \xi, \nabla_V W]=0$ for any $V,W \in \X(\{ S_{\tau} \})$. This allows one to 
define a torsion-free covariant  
derivative $D$ on  $\X(\{ S_{\tau} \})$ by
means of 
\begin{align}
D_V W := \nabla_V W - K(V,W) \xi, \hspace{1cm} \forall V,W \in \mathfrak{X} 
(\{S_\tau\})
\label{GaussIden}
\end{align}
where $K$ is the second fundamental form of $S_{\tau}$ along the transverse normal $\ell$, 
that is
\begin{align}
K(V,W) := - \widehat{\bm{\ell}} \left( \nabla_V W \right), \hspace{1cm} \forall V,W \in \mathfrak{X} (\{ S_\tau\}). %(= (\nabla_V W) (\tau)).
\label{SecFund}
\end{align}
%In both (\ref{GaussIden}) and (\ref{SecFund}) $\nabla_V W$ must be
% understood in the sense of
%(\ref{nabla}).
One has  $D_V W \in \X (\{ S_{\tau} \})$ because (i) 
$\pounds_{\xi} (\nabla_V W ) = \pounds_{\xi} \widehat{\bm{\ell}} = \pounds_{\xi} \xi =0$, hence $\pounds_{\xi} (D_V W) = 0$ and (ii)
$(D_V W) (\tau) = - \widehat{\bm{\ell}} (D_V W) = - \widehat{\bm{\ell}} (\nabla_V W) 
- K(V,W) =0$. By the isomorphism  above, $D$ can also viewed
as a covariant derivative of any of the submanifolds $S_\tau$.
It is immediate to check that this $D$ is actually the Levi-Civita connection
associated to the induced metric 
$$
h(V,W) = g(V,W) , \hspace{1cm} \forall V,W \in \mathfrak{X} (S_\tau).
$$
Note that all $(S_{\tau}, h)$ are isometric for Killing horizons.

Let us introduce the ring of functions $\F(\{ S_{\tau} \})
:= \{ h \in \F(\Hhat) \, ; \,  \xi(h) =0 \}$. It is clear that
$\X (\{S_{\tau}\})$ is a module over $\F(\{ S_{\tau} \})$. Consider its dual
module $\X^{\star} (\{ S_{\tau} \})$, i.e. the set
of $\F(\{ S_{\tau} \})$-linear maps $\bm{\omega} : \X(\{ S_{\tau} \}) 
\longrightarrow \F(\{ S_{\tau}\})$. It is a simple exercise to show that
this module is isomorphic to
\begin{align*}
\mathfrak{X}^* (\{ S_\tau\}) 
:=\{ \bm{\omega}\in \Lambda(\widehat\H), \hspace{3mm} \pounds_\xi \bm{\omega} =0, \hspace{3mm} \bm{\omega} (\xi)=0\},
\end{align*}
and we shall use this representation in the following. The covariant
derivative $D$ extends to the dual $\mathfrak{X}^* (\{ S_\tau\})$ by the standard Leibniz  rule $(D_V \omega)(W) := V (\omega(W)) -\omega( D_V W)$,
where
$V,W \in \X(\{ S_{\tau} \})$.

Let $\bm{\psi}\in T^*_{\Hhat} M$ be any one-form in $M$ along 
$\H$ such that $\pounds_\xi (\iota^\star \bm{\psi})=0$. It follows 
that $\bm{\psi}(\xi) \in \F (\{ S_{\tau} \})$ because
\begin{align*}
\pounds_{\xi} \left ( \bm{\psi} (\xi) \right )
= \pounds_{\xi} \left ( \bm{\psi} (\iota_{\star} (\xi)) \right )
= \pounds_{\xi} \left ( \iota_{\star} (\bm{\psi}) (\xi) \right ) =0.
\end{align*}
Define a $\, \, \widetilde{} \,\, $ operation on such one-forms
$\bm{\psi}$  by
$\widetilde{\bm{\psi}}
:= \iota^\star \bm{\psi} -\bm{\psi} (\xi) d\tau$.  
The property $\pounds_{\xi} \widetilde{\bm{\psi}}=0$ is immediate and,
in addition,
\begin{align*}
\widetilde{\bm{\psi}} (\xi) =
(\iota^\star \bm{\psi}) (\xi)  -\bm{\psi} (\xi) \xi(\tau) =
\bm{\psi} (\iota_{\star} (\xi)) - \bm{\psi} (\xi) =0,
\end{align*}
so $\widetilde{\bm{\psi}} \in \mathfrak{X}^* (\{S_\tau\}) $. A similar
argument establishes 
\begin{align}
\widetilde{\bm{\psi}} (V) = \bm{\psi} (V) \quad \quad 
\forall V \in \X (\{ S_{\tau} \}).
\label{updown}
\end{align}
For any pair
$V,W \in \X (\{ S_{\tau} \})$ we compute
\begin{align}
\left ( \nabla \bm{\psi} \right ) (V,W ) & = \nabla_V ( \bm{\psi} (W))
- \bm{\psi} ( \nabla_V W) 
= 
V (\widetilde{\bm{\psi}} (W) ) - 
 \bm{\psi} \left ( D_V W + K(V,W) \xi \right ) 
\nonumber \\ 
& = 
V (\widetilde{\bm{\psi}} (W) ) - 
 \widetilde{ \bm{\psi}} ( D_V W )
-  K(V,W) \bm{\psi} ( \xi ) \nonumber \\
& = ( D_V \widetilde{\bm{\psi} } ) (W) -
   K(V,W) \bm{\psi} ( \xi ).
\label{Dpsi}
\end{align}
%where (\ref{updown}) was used twice.
% Then as usual one has
%\be
%\nabla \bm{\psi} (V,W) = D \widetilde{\bm{\psi}} (V,W) -\bm{\psi}(\xi) K(V,W) ,% \hspace{1cm} \forall V,W \in \mathfrak{X} (S_\tau). 
%\ee
We want to use this construction applied to $\bm{\Phi}$. 
Observe that $\bm{\Phi}$ satisfies $\pounds_{\xi} ( \iota_{\star} (\bm{\Phi}))=0$ 
because 
\begin{align*}
\left . \begin{array}{ll}
    \forall V \in \X(\{S_{\tau}\}) \quad  &\pounds_{\xi} ( \iota_{\star} (\bm{\Phi})) (V )
= \pounds_{\xi} (  \iota_{\star} (\bm{\Phi}) (V))
= \pounds_{\xi} (\nabla_V \xi) = 0 \\
& \pounds_{\xi} ( \iota_{\star} (\bm{\Phi})) (\xi)
= \pounds_{\xi} ( \bm{\Phi} (\xi) ) = \pounds_{\xi} (\kappa_{\xi} ) =0
\end{array}
\right \} 
\quad 
\Longrightarrow \quad \pounds_{\xi} ( \iota_{\star} (\bm{\Phi} )) =0,
\end{align*}
where we used that $\pounds_{\xi}$ commutes with $\nabla$ and
the implication is a consequence of (\ref{Property}). Thus 
$\widetilde{\bm{\Phi}}$ makes sense and in fact
$$
\widetilde{\bm{\Phi}} = \bm{\varphi} -\k d\tau := -\bm{s}
$$
and using (\ref{wiexact}) 
\be
\iota^\star \bm{w} =d\ln f - \bm{s}\label{s-form}.
\ee
By the isomorphism above,
$\bm{s}$ is actually the torsion one-form of each leaf
$S_{\tau}$:
\begin{align*}
\bm{s}(V) = -\bm{\varphi} (V) = -d\tau (\nabla_V \xi)
=  \iota_{\star} (\bm{\ell} )  \left ( \nabla_V \xi \right )
= \bm{\ell}  \left ( \nabla_V \xi \right )
\quad  \quad \quad \forall V \in \mathfrak{X}(\{ S_\tau\}).
\end{align*}
%\footnote{Note that $\pounds_\xi \bm{\varphi} = d(\bm{\varphi}(\xi)) +i_\xi d\bm\varphi =d(\k) + \xi_\rho \ell^\sigma \iota^\star (i_\xi\bm{\Omega}^\rho{}_\sigma)=0$ where the second equality follows from the formulas in \cite{MS} together with the vanishing of the second fundamental form in $\H$, and the last one from (\ref{Xxi2}).}

Now we can get the essential information contained in (\ref{riemann1}) as well as in (\ref{master}). Let $\{e_A\}$ be a basis of $\mathfrak{X} (\{S_\tau\})$. Then, contraction of (\ref{riemann1}) with $\ell^\mu e^\nu_B $ and letting $X=e_A$ we get, on using (\ref{Dpsi}) and (\ref{nablaxixi}), 
\be
-D_A s_B -\k K_{AB} +s_A s_B \eqHhat e_A^\lambda \xi_\rho R^\rho{}_{\lambda\mu\nu}\ell^\mu e_B^\nu . \label{Ds}
\ee
This is actually an identity valid for any Killing horizon, be it multiple or not. Analogously, setting $X=e_A$ and $Y=e_B$ in (\ref{master}) we arrive at
\be
D_A D_B\ln f +D_A \ln f D_B \ln f -s_A D_B\ln f -s_B D_A \ln f +\k K_{AB}=0\label{master1}
\ee
which is a PDE non-linear in $\ln f$. An alternative form of this PDE, linear in $f$ reads
\be
D_A D_B f -s_A D_B f -s_B D_A f +\k K_{AB} f =0 .\label{master2}
\ee
This is the master equation in neat form. Given any Killing horizon $\H_\xi$ for the Killing vector $\xi$, any other Killing vector sharing the Killing horizon as a degenerate one must satisfy (\ref{etaxi1}) and (\ref{master2}). By using initial-value formulation techniques on null hypersurfaces and bifurcate horizon properties \cite{CCM,CP,CP1,FRW,HIW} one can actually prove that, conversely,
given a solution $f$ of the above equation on any cut of $\H_\xi$ ---and the appropriate initial conditions for the existence of $\xi$ and $\H_\xi$---, there exists a spacetime with a (non-fully degenerate)
%\tim{added}
 MKH for $\xi$ and (\ref{etaxi1}). This will be analyzed in \cite{MPS}.

Expression (\ref{master2}) can thus be seen as a linear system of PDEs for $f$ ---and its trace gives an elliptic PDE on $f$. Given that it  is written in normal form, any solution is determined by the values of $f$ and $D_A f$ at any point $p\in \H$. Therefore, (\ref{master2}) has, at most, $(n-1)+1=n$ independent solutions, which gives the maximum possible dimension for $\Kill_\H^{\deg}$ in accordance with Corollary \ref{coro:dim}. Observe that if $\k =0$ then $f=1$ is one of the solutions and $\xi$ itself is degenerate. 

The precise number of independent solutions that the master equation (\ref{master2}) can have depends on the properties of the ambient spacetime $(M,g)$ and on the intrinsic and extrinsic geometry of the foliation $\{S_\tau\}$ for $\H$ via its integrability conditions. These are briefly derived in the next subsection, and the complete analysis of their consequences will be presented in \cite{MPS}.

\subsection{Integrability conditions}
The integrability conditions of (\ref{master2}) are given by the Ricci identity 
$$
(D_C D_A -D_A D_C) D_B f =-D_D f \stackrel{h}{R}{}^D{}_{BCA}
$$
where $\stackrel{h}{R}$ is the curvature tensor of the connection $D$, which coincides with the Riemann tensor of any of the cuts $(S_\tau,h)$. A straightforward calculation using (\ref{master2}) leads to
\bea
D_D f\left[\stackrel{h}{R}{}^D{}_{BCA}-\delta^D_B \left(D_A s_C -D_C s_A \right) +s_B\left(s_A \delta^D_C -s_C \delta^D_A \right)-\delta^D_C D_A s_B +\delta^D_A D_Cs_B \right. \nonumber\\
\left. +\k \left(\delta^D_A K_{CB} -\delta^D_C K_{AB} \right) \right]
+\k f\left(D_A K_{CB} -D_C K_{AB} -s_A K_{CB} +s_C K_{AB}Ê\right)=0  \label{intcond}
\eea
which can be rewritten, on using (\ref{Ds}), as
\bea
D_D f\left[\stackrel{h}{R}{}^D{}_{BCA}-\delta^D_B \left(D_A s_C -D_C s_A \right) +
(\delta^D_C e_A^\lambda -\delta^D_A e_C^\lambda )\xi_\rho R^\rho{}_{\lambda\mu\nu}\ell^\mu e_B^\nu \right]\nonumber \\
+\k f\left(D_A K_{CB} -D_C K_{AB} -s_A K_{CB} +s_C K_{AB}Ê\right)\eqHeta 0 .\label{intcond1}
\eea
Using here for the last term in brackets the Codazzi equation for the foliation $\{S_\tau\}$ we can still write
$$
D_D f\left[\stackrel{h}{R}{}^D{}_{BCA}-\delta^D_B \left(D_A s_C -D_C s_A \right) +
(\delta^D_C e_A^\lambda -\delta^D_A e_C^\lambda )\xi_\rho R^\rho{}_{\lambda\mu\nu}\ell^\mu e_B^\nu \right]
+\k f \ell_\rho R^\rho{}_{\lambda\mu\nu} e^\lambda_B e^\mu_C e^\nu_A \eqHeta 0.
$$
Every MKH lives in a spacetime such that this is satisfied by the function $f$ in (\ref{etaxi}), for $\eta \in \Kill^{deg}_\H$. In particular, the maximum dimension of $\Kill^{\deg}_\H$ is attained whenever the previous condition holds identically, that is, for any values of $f$ and $D_A f$. In other words, when the term in brackets vanishes and the factor multiplying $f$ does too. This allows us to analyze in detail the spacetimes with (fully degenerate or not) MKHs of maximal order, as well as the cases with other values of the order $m$, see \cite{MPS}.

\section{Classification of MKHs in maximally symmetric spacetimes}\label{App:MaxSym}
%\jose{This appendix is rather long. Also, I have not included some Lemmas that might needed. I have actually not included some long proofs...  Any ideas?}
%\marc{Restructured}

In this section we study the multiple Killing horizons in the (A)-de Sitter and Minkowski spacetimes of arbitrary dimension $n+1$ at least two. 
Among other things, we show that any point $p$ in these spacetimes is contained in a multiple Killing horizon of maximal order
$n+1$. We start with the A-dS case, which requires a machinery that can then be applied to the Minkowski case.

\subsection{The (A)dS case}
The (anti)-de Sitter space of curvature radius $a >0$,
denoted by $\ads^{n+1}_a$,  is the
 maximally extended and simply connected $(n+1)$-dimensional
($n \geq 1$) Lorentzian manifold
of constant  curvature $K = \frac{\epsilon}{a^2}$ where $\epsilon =1$
in the de Sitter case and $\epsilon =-1$ in the anti-de Sitter case.
We intend to give the full classification of MKHs in these
spaces. From theorem \ref{th:algebra} we know that any such MKH has dim $\Kill^{deg}_\H =m-1$, where $m\geq 2$ is the order of the MKH, so that
to classify the MKHs it suffices to determine all degenerate Killing horizons, and then
find which of those are multiple. 

To that aim, it is convenient to view $\ads_a$ as an embedded hypersurface
in a higher-dimensional flat space. More specifically,
let $\Minkpq$ be the simply connected, complete 
pseudo-Riemannian manifold of vanishing curvature and signature
$(p,q)$. We assume  $p+q = n+2$ and select a Cartesian coordinate system
$\{ x^{\alpha'} \}$ ($\alpha',\beta'
\cdots = 0,\cdots,n+1$) which will stay fixed from now on.
The components of the flat metric $g^\flat$ in these coordinates are
$g^\flat_{\alpha'\beta'} 
= \mbox{diag} \{ \underbrace{-1, \cdots, -1}_{p}, \underbrace{+1,\cdots, +1}_{q}\}$. 

We shall consider the two cases at the same time. Recall that $\epsilon := \pm 1$,
and fix   $2 p = 3- \epsilon$, i.e. when $\epsilon =1$ we work 
with the $(n+2)$-dimensional Minkowski space 
$\M^{1,n+1}$ and when $\epsilon = -1$ we have $\M^{2,n}$. Denote them collectively by $\M_\epsilon^{n+2}$. There exists an isometric immersion of 
$\ads_a^{n+1}$ into $\M_\epsilon^{n+2}$ whose image is
\begin{align*}
\Sigma_{a} := \{ x \in \M_{\epsilon}^{n+2}, \quad  \la x,x\ra= \epsilon a^2 \},
\end{align*}
where $\la \, , \, \ra$ denotes scalar product with $g^\flat$
and we are making use of the affine structure of $\M_{\epsilon}^{n+2}$,
which makes it into a vector space  with inner
product $g^\flat$. 
When $\epsilon =1$ the immersion is in fact a proper embedding. When
$\epsilon =-1$, there is a covering map $\pi : \mbox{AdS}_a^{n+1} \longrightarrow
\widetilde{\mbox{AdS}\,}{}_a^{n+1}$ onto a space which is diffeomorphic
to $\Sigma_a$. Thus, the MKHs in $\ads_a^{n+1}$ can be studied by considering their images in $\Sigma_a$.

The algebra of Killing vectors of $\ads_a^{n+1}$ can be obtained by restriction
in $\Sigma_a$ of the set of Killing vectors in $\M_{\epsilon}^{n+2}$ which 
leave the origin $ o \in \M_{\epsilon}^{n+2}$ invariant, given by
\begin{align*}
\zeta_F|_x = F^{\sharp}(x)
\end{align*}
where $F^{\sharp} := \M_{\epsilon}^{n+2} \longrightarrow \M_{\epsilon}^{n+2}$ is a skew
symmetric linear map, i.e. satisfying for all $x,y \in \M_{\epsilon}^{n+2}$
\begin{align*}
\la F^{\sharp} (x),y \ra = - \la x, F^{\sharp}(y) \ra.
\end{align*}
Given a non-zero vector $Z \in \M_{\epsilon}^{n+2}$  we define
$\la Z\ra^{\perp}$ to be the hyperplane orthogonal to $Z$ passing through the origin, i.e.
the set of points $\{ x \in \M_{\epsilon}^{n+2}: \la Z, x \ra =0\}$.
A point $Z \in \M_{\epsilon}^{n+2}$ is called respectively timelike, null or spacelike 
if $\la Z,Z\ra$ is negative, zero or positive.

We can now state our main result concerning degenerate
Killing horizons in $\ads_a^n$.

\begin{theorem}
\label{degenerate}
Let $\H$ be a degenerate Killing horizon of $\ads_a^{n+1}$. Then there exists
a null, non-zero vector $k \in \M_{\epsilon}^{n+2}$ such that 
$\H$ is a subset of the intersection of $\Sigma_a$ with the hyperplane 
$k^{\perp} \subset \M_{\epsilon}^{n+2}$. Moreover, the set of Killing vectors 
which respect to which an open and dense subset of $\H$
is a degenerate Killing horizon is given by the restriction to 
$\Sigma_a$ of $\zeta_F$ with
\begin{align*}
F^{\sharp} = k \otimes \bm{w} - w \otimes \bm{k}
\end{align*}
and $w \in \M_{\epsilon}^{n+2}$ is a vector linearly independent
of $k$ and satisfying $\la k,w \ra =0$. 
Conversely, for any pair $\{k,w\}$ as before, the Killing vector 
$\zeta := \zeta_{F^{\sharp}} |_{\Sigma_a}$ admits a degenerate Killing horizon given by
the hypersurface 
\begin{align*}
\H_{\zeta} := \{ x \in \Sigma_a \cap \la k \ra^{\perp} \mbox{ such that }
\la w, x \ra \neq 0 \}
\end{align*}
or any open subset thereof.
\end{theorem}
\begin{remark}
When $\epsilon =1$, $w$ must be spacelike because a causal vector perpendicular to $k$ cannot be linearly independent of $k$.
When $\epsilon =-1$, there is no such restriction and
$w$ is allowed to have any norm (including null).
\end{remark}

The proof of this theorem is somewhat long, and requires several results on skew symmetric linear
maps on pseudo-Riemannian vector spaces. We devote Appendix \ref{ApB} to  establishing the necessary lemmas
and give the proof.

%\jose{Available proof should be shortened?} 
%\marc{Added in an Appendix}

With this theorem above at hand, it is easy to determine the MKHs in $\ads_a^{n+1}$.
\begin{theorem}
Let $\ads_a^{n+1}$ be the (A)-de Sitter spacetime of 
dimension $n+1 \geq 2$ and view this as a hypersurface in $\M_{\epsilon}^{n+2}$ as
described above.  A null hypersurface $\H$ embedded in $\ads_a^{n+1}$
is a multiple Killing horizon if and only
if $\H$ is an open subset of the hypersurface
$\la k \ra ^{\perp} \cap \Sigma_{a}$ where $k \in \M_{\epsilon}^{n+2}$
is non-zero and null. Moreover, $\Kill_\H$ is generated
by the restriction to $\Sigma_a$ of the Killing vectors in $\M_{\epsilon}^{n+2}$ 
$\zeta_{F_{k,Z}} (x) = F_{k,Z}^{\sharp} (x)$, $x \in \M_{\epsilon}^{n+2}$ with $F^{\sharp}_{k,Z} := \M_{\epsilon}^{n+2}
\longrightarrow \M_{\epsilon}^{n+2}$ given by
\begin{align}
\quad F_{k,Z}^{\sharp}: = k \otimes \bm{Z} - Z \otimes \bm{k}, \quad \quad 
Z \in \M_{\epsilon}^{n+2}.
\label{FormFkv}
\end{align}
\end{theorem}
\begin{remark}
The collection of vectors $\zeta_{F_{k,Z}}$ is obviously a vector
subspace of all Killing vectors in $\M_{\epsilon}^{n+2}$ leaving invariant the origin 
%\marc{several typos $\M_{\epsilon} \rightarrow \M_{\epsilon}^{n+2}$ corrected until the end of the file}
of  $\M_{\epsilon}^{n+2}$, in agreement with theorem \ref{KillH}. Define 
the equivalence relation, $Z \sim Z' \Leftrightarrow
Z - Z' \in \mbox{span} (k) $. The quotient space, denoted 
$\M_{\epsilon}^{n+2}/k$, is clearly an $(n+1)$-dimensional vector space. It turns out that 
$\Kill_\H$ is isomorphic to $\M_{\epsilon}^{n+2} /k$. Indeed, define the map
\begin{align*}
  \Psi: \M_{\epsilon}^{n+2}/k & \longrightarrow \Kill_{\H} \\
\overline{Z} & \longrightarrow \zeta_{F_{k,Z}} |_{\Sigma_a}
\end{align*}
where $Z$ is any representative in the equivalence class $\overline{Z}$.
This map is well defined (i.e. independent of the representative chosen
in the class) because for $Z' = Z + c k$, $c \in \mathbb{R}$, 
\begin{align*}
F^{\sharp}_{k,Z'} = k \otimes \bm{Z'} - Z' \otimes \bm{k}
= k \otimes (\bm{Z} + c \bm{ k}) - (Z + c k) \otimes \bm{k} =
 k \otimes \bm{Z} - Z \otimes \bm{k}
= F^{\sharp}_{k,Z}
\end{align*}
and the Killing vector  $\zeta_{F_{k,Z'}} = \zeta_{F_{k,Z}}$. The map is obviously
linear. It is also a bijection because 
$\zeta_{F_{k,Z}}  = \zeta_{F_{k,Z'}}$ agree on $\Sigma_a$ if and
only if they agree everywhere, i.e.
$F^{\sharp}_{k,Z} = F^{\sharp}_{k,Z'}$ or, explicitly,
\begin{align}
k \otimes \bm{(Z' -Z)} - (Z'-Z) \otimes \bm{k} =0.
\label{cond}
\end{align}
This clearly holds if and only if $Z'-Z$ proportional to $k$. 
%(if they were not 
%proportional their normal spaces would not be  same, so there would exist
%a vector  $u \in \la k \ra^{\perp}$ not lying in $\la v'-v \ra^{\perp}$. Applying 
%the endomorfism (\ref{cond}) 
%to $u$,  a contradiction would follow).
We therefore conclude that the dimension of $\Kill_{\H}$ is $n+1$.
\end{remark}
\begin{proof}
$\H$ has an open and dense subset $\H_{\zeta}$ which is a degenerate Killing
horizon of $\ads_a^{n+1}$ associated to the Killing vector $\zeta$. 
By theorem \ref{degenerate}, this occurs if and only if
there exists $k \in \M_{\epsilon}^{n+2}$
null and non-zero such that $\H$ is an open subset of 
$\Sigma_a \cap \la k \ra^{\perp}$. This proves the first part of
the theorem.

In order to identify $\Kill_\H$,  let $\H_{\xi}$ be a Killing horizon 
(not necessarily degenerate) such that 
$\overline{\H_{\xi} } = \overline{\H_{\zeta}} = \overline{\H}$. Since 
$\Sigma_a \cap \la k \ra^{\perp}$ is closed,  we also have 
$\overline{\H_{\xi}} \subset \Sigma_a \cap \la k \ra^{\perp}$. Let
$F_{\xi}^{\sharp}$ be the endomorphism in $\M_{\epsilon}^{n+2}$ such that
$\xi|_x = F_{\xi}^{\sharp}(x)$, $\forall x \in \Sigma_a$. 
Up to scaling, $k$ is the only normal to $\H_{\xi} \subset \Sigma_a$.
Thus, it must be that at any point $x \in \H_{\xi}$, 
$F_{\xi}^{\sharp} (x) = Z|_x k$ holds, where $Z|_x$ is a non-zero
real number (it may depend on $x \in \H_{\xi}$). Since $\H_{\xi}$ is an
open subset of $\Sigma_a \cap \la k \ra^{\perp}$, it follows that
$\mbox{span}( \H_{\xi} ) = \la k \ra^{\perp}$.
By linearity of $F_{\xi}^{\sharp}$ it follows
\begin{align}
F^{\sharp}_{\xi} (w) = Z|_w k, \quad \quad \forall w \in  \la k\ra^{\perp}.
\label{Fk}
\end{align}
We may apply Lemma \ref{Fkperp} to conclude
that $F_{\xi}^{\sharp}$ is given as
in (\ref{FormFkv}) and hence any
$\xi \in \Kill_\H \setminus \{ \bm{0} \}$ must be 
the restriction to $\Sigma_a$ of $\zeta_{F_{k,Z}}$, as claimed in the theorem.
Conversely, 
any $F^{\sharp}_{k,Z}$ of this form with $Z$ and $k$ linearly independent defines
a Killing vector in $\M_{\epsilon}^{n+2}$ which, when restricted to 
$\Sigma_a \cap  \la k \ra^{\perp}$ gives a null, tangent vector.
Combined with the fact that  when $Z = \alpha k$, $\alpha \in \mathbb{R}$ 
we have $F^{\sharp}_{k,Z} =0$ and therefore $\zeta_{F_{k,Z}} =0$ we conclude that
\begin{align*}
\Kill_\H = \{ \zeta_{F_{k,Z}} |_{\Sigma_a} \}
\end{align*}
and the theorem is proved.
\end{proof}

\begin{remark}
The Killing horizon of $\zeta_{F_{k,Z}} |_{\Sigma_a}$ is (any open
subset of)
\begin{align*}
  \H_{k,Z} = \{ x \in \Sigma_a \cap \la k \ra^{\perp} \mbox{ such that }
\la Z, x \ra \neq 0\}.
\end{align*}
To compute the surface gravity we first note that the square norm
of $\zeta_{F_{k,Z}}$ is
\begin{align*}
\la Z, Z \ra \la k,x \ra^2 - 2 \la k,Z \ra
\la k,x \ra \la Z,x \ra
\end{align*}
whose  gradient evaluated at $x \in \H_{k,v}$ reads
\begin{align*}
 2 \la k ,Z \ra \la Z, x \ra k.
\end{align*}
Given that (at such $x$) $\zeta_{F_{k,Z}} |_x = k \la Z, x \ra$
we conclude from (\ref{SurGrav}) 
\begin{align*}
\kappa_{\H_{k,Z}} = \la k , Z \ra.
\end{align*}
Note that when $Z$ and $k$ are orthogonal, the surface gravity is zero and
we recover the degenerate Killing horizon of theorem \ref{degenerate}.
\end{remark}

\subsection{The Minkowski case}
Using the same notation as above, the general Killing vector $\zeta$ of $\Mink$ is
\begin{align}
\zeta_{z,F^{\sharp}} |_x  = z + F^{\sharp}(x)
\label{KVMink}
\end{align}
where $z \in \Mink$, $F^{\sharp}: \Mink \longrightarrow \Mink$ is a
skew-symmetric endomorphism.
As in the previous subsection, we start with the the classification
of degenerate Killing horizons. 

\begin{theorem}
\label{DegMink}
Let $\H$ be a degenerate Killing horizon of a Killing vector $\zeta$ in
$\Mink$. Then, and only then, one of the two following possibilities hold:
\begin{itemize}
\item[(a)] There exists $z,z' \in \Mink$ with $z$ null and non-zero
such that
$\zeta = z$ and $\H$ is an open subset of the hyperplane $\H_{z',z} := z'
+ \la z \ra^{\perp}$.
\item[(b)] There exist $A \in \mathbb{R}$ and 
$k, w, z' \in \Mink$ with $\{ k,w\}$ linearly independent,
$k$ null, $w$ spacelike and orthogonal to $k$ and $z'$ arbitrary,
such that
\begin{align}
\zeta |_x = A k + k \la w, x \ra - w  \la k, x -z' \ra
\label{formxi}
\end{align}
and  $\H$ is an open subset of the hypersurface 
\begin{align}
(z' + \la k \ra^{\perp}) \setminus S_w
\label{KillHorMin} 
\end{align}
where $S_w$ is the closed, codimension-two null plane defined by
\begin{align}
S_w := - A \la w,w\ra^{-1} w + \mbox{span} (k,w)^{\perp}.
\label{defSw}
\end{align}
\end{itemize}
\end{theorem}

\begin{proof}
Let $\lambda_{\zeta} := - g^\flat(\zeta,\zeta)$.  The degenerate
Killing horizon  $\H$ must be  a subset of 
$\{ \lambda_{\zeta} =0\} \cap \{ \mbox{grad} (\lambda_{\zeta})=0 \}$. 
Let $z$, $F^{\sharp}$ be such that
$\zeta = \zeta_{z,F^{\sharp}}$.   The square norm of $\zeta$ is
\begin{align*}
- \lambda_{\zeta} = \la z, z \ra + 2 \la x, F^{\sharp} (x) \ra
 \la F^{\sharp}(x), F^{\sharp}(x) \ra
\end{align*}
and the gradient
\begin{align*}
- \mbox{grad} (\lambda_{\zeta}) = - 2 F^{\sharp} (z + F^{\sharp}(x)).
\end{align*}
This implies that for any 
$x_1, x_2 \in \H$
\begin{align*}
(F^{\sharp} \circ F^{\sharp}) (x_1 - x_2 ) = 0
\end{align*}
holds, so $x_1 -x_2$ belongs to the kernel of $F^{\sharp} \circ F^{\sharp}$. In
particular, the tangent space $T_x \H$ at any $x \in \H$ must satisfy
\begin{align*}
T_x \H \subset \mbox{Ker} (F^{\sharp} \circ F^{\sharp} ).
\end{align*}
Since $T_x \H$ is $n$-dimensional it must be 
that $\mbox{dim} (\mbox{Ker} (F^{\sharp} \circ F^{\sharp} ))$ is either
$n$ or $n+1$.
In the latter case, called (a) Lemma \ref{FsqLemma} implies $F^{\sharp}=0$, so that  $\zeta = z$ with $\la z, z
\ra =0$. 
Thus,  $\zeta$ is null and non-zero everywhere  and
$\Mink$ is foliated by Killing horizons of $\zeta$ defined as the
hypersurfaces orthogonal to $z$, i.e.
the hyperplanes $z' + \la z \ra^{\perp}$, $z' \in \Mink$. This proves 
case (a) of the theorem.

Consider next case (b), defined by the condition that
$F^{\sharp} \circ F^{\sharp}$ has rank one, 
or equivalently, there is $k \in \Mink$ non-zero such that
$F^{\sharp} \circ F^{\sharp} = \mu k \otimes \bm{k}$, $\mu \neq 0$.
The kernel of $F^{\sharp} 
\circ F^{\sharp}$ (namely $\la k \ra^{\perp}$)  must contain the null hyperplane
$T_x \H, x \in \H$, so $k$ must be null
and   $T_x \H = \la k \ra^{\perp}$ for all $x \in \H$. Thus, $\H$ must be a subset of one of the hyperplanes
normal to $k$. In other words, there is $z` \in \Mink$ such that
$\H$ is an open subset of 
$\H_{z',k} := z' + \la k \ra^{\perp}$.  To impose the condition
that $\zeta$ is null and tangent to $\H$, we need the form of $F^{\sharp}$.
We apply Lemma \ref{rank1} and 
find that there exists
$w \in \Mink$, orthogonal to, and linearly 
independent of, $k$ such that
\begin{align*}
F^{\sharp} = k \otimes \bm{w} - w \otimes \bm{k}.
\end{align*}
Note that in Lorentzian signature $w$ is necessarily spacelike, so 
$\mu  = - \la w, w \ra <0$. Evaluating $\zeta$ at  $x \in 
\H \subset \H_{z',k}$ one finds
\begin{align*}
\zeta |_x = z + F^{\sharp} (x) = z + k \la w, x \ra - w \la k, x \ra
= z + k \la w, x \ra - w \la k,z' \ra.
\end{align*}
This vector is proportional to the normal of $\H$ (i.e. to $k$) if and only if
$z = w \la k,z' \ra + A k$ for some $A \in \mathbb{R}$. This shows (\ref{formxi}).
To prove  (\ref{KillHorMin}) we simply note that 
$\zeta|_x$, as given in (\ref{formxi})  vanishes at $x \in \H_{z',k}$ 
if and only if   $A + \la w,x \ra =0$. Write $x = - A \la w,w\ra^{-1} w + y$
and this condition becomes $\la w, y \ra =0$, as claimed in the proposition.
The ``only then'' part in case (b) is immediately checked.
\end{proof}

We can now classify the MKHs in the
$(n+1)$-dimensional Minkowski spacetime.

\begin{theorem}
Let $\Mink$ be the Minkowski spacetime of 
dimension $n+1 \geq 2$.
A null hypersurface $\H$ embedded in $\Mink$
is a multiple Killing horizon if and only
if $\H$ is an open subset of a hyperplane
$z' + \la k \ra ^{\perp}$ with $z', k \in \Mink$ and 
$k$ is null and non-zero.  Moreover, $\Kill_\H$ is given by
\begin{align*}
\Kill_\H = \{ \zeta_{A,Z}|_x =  A k + Z \la k,z' \ra +  k \la Z,x \ra - 
Z \la k, x \ra, \quad \quad A \in \mathbb{R}, Z\in \Mink \}.
\end{align*}
\end{theorem}

\begin{proof}
$\H$ has an open and dense subset $\H_{\zeta}$ which is a degenerate
Killing horizon of $\Mink$.
By Theorem \ref{DegMink} we know that $\H_{\zeta}$ is an open subset of a hyperplane 
$\H_{z',k} := z' + \la k \ra^{\perp}$ where $k \neq 0$ is null.  To show that 
$\H$ is a multiple horizon (and also to determine $\Kill_{\H}$)
we need to find the most general Killing vector
$\zeta_{z,F^{\sharp}}$ admitting a Killing horizon, denoted by
$\H_{z,F^{\sharp}}$  such that
$\overline{\H_{z,F^{\sharp}}} = \overline{\H_{\zeta}} = \overline{\H}$. 
$\H_{z,F^{\sharp}}$ is an open subset of $z' + \la k \ra^{\perp}$, 
so the condition that $\zeta$ is null and tangent to $\H_{z,F^{\sharp}}$
on $\H_{z,F^{\sharp}}$, namely
\begin{align*}
\zeta_{z,F^{\sharp}} |_x = z + F^{\sharp} (x)  = f|_x k, \quad \quad \forall 
x \in \H_{z,F} 
\end{align*}
must hold. By linearity this relation extends
to all $c + \la k \ra^{\perp}$. Thus, for all
$X \in \la k \ra^{\perp}$ 
\begin{align*}
F^{\sharp} (X) = f|_X k - F^{\sharp}(z') -  z
\end{align*}
holds.
This applies, in particular to
$X=0$ from which $z = f|_0 k - F^{\sharp}(z')$
and thus
\begin{align*}
F^{\sharp} (X) = (f|_X - f|_0) k, \quad \quad \forall X \in \la k \ra^{\perp}.
\end{align*}
which allows us to conclude that there is $Z \in \Mink$
such that
\begin{align*}
F^{\sharp} = k \otimes \bm{Z} - Z \otimes \bm{k}.
\end{align*}
Note that this implies $z = f|_0 k - F^{\sharp} (z') = f|_0 k - k \la Z, z'\ra 
+ Z \la k,z' \ra = A k + Z \la k,z' \ra$, after redefining $A:= f|_0 
- \la Z, z' \ra$. We have proved the inclusion 
\begin{align*}
\Kill_{\H} \subset \{ \zeta|_x =  A k + Z \la k,z' \ra +  k \la Z,x \ra - 
Z \la k, x \ra, \quad \quad A \in \mathbb{R}, Z \in \Mink \}.
\end{align*}
The reverse  inclusion (and hence equality)
is immediate, since the Killing  vector $\zeta_{A,Z}$ (with obvious notation)
 is tangent and null at the hyperplane $z' + \la k \ra^{\perp}$
and vanishes only on the lower dimensional subset
\begin{align*}
S_{A,Z} := \{ Y \in  z' + \la k \ra^{\perp} ; A + \la Z,Y\ra =0 \}.
\end{align*}
\end{proof}

\begin{remark}
Two Killing vectors $\zeta_{A,Z}$ and $\zeta_{A',Z'}$ agree iff and only if $Z' - 
Z = a k$ and $A' = A - a \la k ,z' \ra$, for some arbitrary constant
$a$.  Thus, the dimension
of $\Kill_\H$ is $n+1$.
%\marc{This raises the question: an $n$ dimensional spacetime admitting
%a multiple Killing horizon $\H$ with $\mbox{dim} (\Kill_M) =n$ must be 
%locally Minkowski or (A)-de Sitter? Maybe for $n$ larger than three?
%Maybe for Einstein spaces? } 
\end{remark}

The surface gravity of the Killing horizon associated to
$\zeta_{A,Z}$ is computed easily as follows
\begin{align*}
- \lambda_{A,Z} := g^\flat (\zeta_{A,Z},\zeta_{A,Z}) 
=  \la Z, Z \ra \la k , z' - x \ra^2
 + 2 \la k, Z \ra \left ( A + \la Z, x \ra \right ),
\la k , z'- x\ra 
\end{align*}
so its gradient is
\begin{align*}
\mbox{grad} (\lambda_{A,Z} )
=  2 \la Z, Z \ra \la k, z'- x \ra k
- 2 \la k, Z \ra \la k , z'- x \ra Z + 
2 \la k ,Z \ra  \left ( A + \la Z, x \ra \right ) k,
\end{align*}
which evaluated on $z' + \la k \ra^{\perp}$ gives
\begin{align*}
\mbox{grad} (\lambda_{A,Z} ) |_{z' + \la k \ra^{\perp}} 
= 2 \la k ,Z \ra \left ( A + \la Z, x \ra \right ) k
= 2 \la k, Z \ra \zeta_{A,Z} |_{z'+ \la k \ra^{\perp}}
\end{align*}
and the surface gravity is $\kappa_{A,Z} = \la k, Z \ra$.

\section*{Acknowledgments}%\jose{I have used those of our last paper, please check...}
MM acknowledges financial support under projects
FIS2015-65140-P (Spanish MINECO/FEDER) and
SA083P17 (Junta de Castilla y Le\'on).
TTP acknowledges financial support by the Austrian Science Fund (FWF)
P~28495-N27.
JMMS is supported under Grants No. FIS2017-85076-P (Spanish MINECO/AEI/FEDER, UE) and No. IT956-16 (Basque Government).

\appendixpage
\appendix

%\section*{Appendices}

%\marc{Title changed}
\section{Lower bound on co-dimension of fixed-point sets for Killing vectors}\label{app:1}
Here we recall the following well-known fact, which we nevertheless
prove for completeness. A Killing vector
is {\bf non-trivial} if it is not the zero vector field.
\begin{lemma}
\label{codimension_two}
Let $(M,g)$ be an $(n+1)$-dimensional spacetime and $\xi$ a non-trivial Killing
vector. Then the set of zeros of $\xi$ has co-dimension at least two.
\end{lemma}
\begin{proof}
We consider the relevant case $n \geq 1$. We know that the zeros of a Killing vector
form a finite collection of smooth embedded submanifolds ${\mathcal S_i}$
\cite{Kobayashi}. Let $p$ be a point in one of them, say ${\mathcal S}_1$ and
assume that $\mbox{dim} ({\mathcal S}_1) \geq  n$. Let $G^{\sharp}$ be the
endomorphism $T_p M \longrightarrow T_p M$ defined by
$g(G^{\sharp}(Z),Z') =  d \bm{\xi} (Z,Z'), \hspace{2mm} \forall Z,Z'\in T_pM$. Since 
$d \bm{\xi} |_p$ is a two-form in $T_p M$, $G^{\sharp}$ is skew symmetric 
with respect to $g|_p$. The tangent plane $T_p {\mathcal S}_1$
lies in the kernel of $G^{\sharp}$, so its dimension is at least $n$,
%By Lemma \ref{Ker} in Appendix \ref{ApA},
or equivalently rank $(G^\sharp)\in\{0,1\}$. If rank $G^\sharp =1$ then $G^\sharp = k \otimes \bm{a}$ for some vector $k\in T_pM$ and some one-form $\bm{a}\in \Lambda_p M$, which is clearly incompatible with the skew-symmetry of $G^\sharp$ ---as $g$ is non-degenerate--- unless $\bm{a}=0$.
%Skew symmetry then implies
%$$
%g(G^\sharp(Z),Z) = g(k,Z) \bm{a} (Z) =d\bm{\xi} (Z,Z) =0, \hspace{3mm} \forall Z\in T_p M
%$$
%so that $\bm{a}(Z)=0$  for every $Z$ which is not orthogonal to $k$. 
Thus, $G^{\sharp} =0$. i.e. $d\bm{\xi}|_p =0$. This immediately implies that $\xi$ is a trivial Killing vector.
\end{proof}

The previous theorem can be considered to hold for $n=0$ too if the statement is understood as saying that
$\xi$ cannot have zeros. For assume
$p \in M$ were a fixed point of
$\xi$ and select a coordinate chart  $\{ x\}$ containing $p$, with $x_p := x(p)$.
The metric can be written as $g = j(x) dx^2$, with
$j$ non-zero in the domain of the chart. The Killing could be written as $\xi 
= l (x) \partial_x$ with $l(x_p)=0$. 
The condition of being a Killing vector is
\begin{align*}
\pounds_\xi( g) = 0 \quad \quad \Longleftrightarrow \quad \quad
 l \frac{dj}{dx} - 2 j \frac{dl}{dx} =0.
\end{align*}
Since $l(x_p)=0$, uniqueness of solutions of ODE would imply $l(x) = 0$ everywhere, so the Killing would be trivial.

\section{Proof of Theorem \ref{degenerate}}
\label{ApB}

In order to prove Theorem \ref{degenerate}
we need several  algebraic lemmas
on skew symmetric linear maps. 
Several of these results are likely to be known in the 
mathematics literature, but they are not standard
knowledge in the relativity community (given that they involve
various signatures). So we provide a proof
for completeness.

\begin{lemma}
\label{TotallyDeg}
Let $(V, g^{\flat})$ be an $n$-dimensional vector space and $g^{\flat}$ a pseudo-riemannian
inner product of signature $\{p,q\}$. Let  
$\Pi$ be a linear subspace with the property that $g^{\flat}$ restricted to $\Pi$
is identically zero (we call such spaces {\bf totally degenerate}).
Then the dimension of $\Pi$ is bounded above by
$\mbox{min}(p,q)$, and this bound is sharp.
\end{lemma}

\begin{proof}
By interchanging $g^{\flat}$ with $-g^{\flat}$, we may assume without loss of generality
that $p \leq q$. Let $\{e_i\}$ by an orthonormal basis of $(V,g^{\flat})$ and
consider the vector space $\Pi_0 = \mbox{span} ( e_1 + e_{p+1}, e_2 + e_{p+2},
\cdots, e_{p} + e_{2p})$, which has dimension $p$. Since $\la e_i + e_{p+i}, e_j + e_{p+j} \ra
= \la e_i, e_j \ra + \la e_{p+i}, e_{p+j} \ra = - \delta_{ij} + \delta_{ij} =0$,
the restriction $g^{\flat} |_{\Pi_0}$ is identically zero. Thus,
the upper bound claimed in the lemma is attained.

It remains to
show that  any totally degenerate vector subspace $\Pi$ satisfies
$\mbox{dim} (\Pi ) \leq p$. We argue by contradiction, so let $\Pi$ by a
totally degenerate space of dimension  $p+1$ and $\{ v_1,
\cdots, v_{p+1}\}$ a basis of  $\Pi$.  
The orthogonal decomposition $V =
\mbox{span} \{ e_1, \cdots e_p\} \oplus \{ e_{p+1}, \cdot, e_{p+q} \}$ 
allows us to  decompose any $v \in V$ as
 $v = v^{\parallel} + v^{\perp}$. It is clear that 
$\{ v_1^{\parallel}, \cdots, v_{p+1}^{\parallel} \}$ is a linearly dependent subset.
By reordering vectors if necessary we may assume that
$v_{p+1}^{\parallel} = \sum_{i=1}^{p} a_i v_i^{\parallel}$. The fact that $g^{\flat}|_{\Pi} =0$ implies, for all $a,b = 1, \cdots, p+1$, 
\begin{align}
0 = \la v_a, v_b \ra = \la v_a^{\parallel} + v_a^{\perp},
v_b^{\parallel} + v_b^{\perp} \ra = 
\la v_a^{\parallel}, v_b^{\parallel} \ra + \la v_a^{\perp}, v_b^{\perp} \ra 
\quad \Longleftrightarrow \quad 
\la v_a^{\perp}, v_b^{\perp} \ra  = - \la v_a^{\parallel}, v_b^{\parallel} \ra. 
\label{orto}
\end{align}
Let us compute
\begin{align*}
\Big \la v^{\perp}_{p+1}  - \sum_{i=0}^{p} a_i v_i^{\perp}, 
v^{\perp}_{p+1} - \sum_{i=0}^{p} a_i v_i^{\perp} \Big \ra & = 
\la v^{\perp}_{p+1} , v^{\perp}_{p+1} \ra
- 2 \sum_{i=1}^p a_i \la v_{p+1}^{\perp}, v_i^{\perp} \ra
+ \sum_{i=1}^{p} \sum_{j=1}^{p} a_i a_j \la v_i^{\perp}, v_j^{\perp} \ra \\
& = 
- \la v^{\parallel}_{p+1} , v^{\parallel}_{p+1} \ra
+ 2 \sum_{i=1}^p a_i \la v_{p+1}^{\parallel}, v_i^{\parallel} \ra
- \sum_{i=1}^{p} \sum_{j=1}^{p} a_i a_j \la v_i^{\parallel}, v_j^{\parallel} \ra \\
& = -
\Big \la v^{\parallel}_{p+1} - \sum_{i=0}^{p} a_i v_i^{\parallel}, 
v^{\parallel}_{p+1} - \sum_{i=0}^{p} a_i v_i^{\parallel} \Big \ra =  0,
\end{align*}
where in the third equality we used (\ref{orto}). Since
$v^{\perp}_{p+1} - \sum_{i=0}^{p} a_i v_i^{\perp}$ lies in a $q$-dimensional vector
subspace where $g^{\flat}$ is positive definite it must be
$v^{\perp}_{p+1} = \sum_{i=0}^{p} a_i v_i^{\perp}$, but then
also $v_{p+1} = \sum_{i=0}^{p} a_i v_i$, which is a contradiction.
\end{proof} 

We shall also need the following property of totally degenerate subspaces of maximal dimension.  
\begin{lemma}
\label{Ortogonality}
Let $(V, g^{\flat})$ satisfy the same assumptions as in Lemma \ref{TotallyDeg}.
Let $\Pi$ be a totally degenerate vector subspace of maximal dimension
$r:= \mbox{min}(p,q)$ and $\{ k_1, \cdots, k_r\}$ a basis of $\Pi$. Select
any $r$-dimensional vector subspace $T$ with the property that $g^{\flat} |_T$
is negative definite (if $p \leq q$) or positive definite
(if $p \geq q$) and for any $v \in V$ write 
$v= v^{\parallel} + v^{\perp}$ 
according to the direct sum decomposition $V = T \oplus T^{\perp}$.
% where $T^{\perp}$ denotes the orthogonal space of $T$. 
Then the following properties
hold:
\begin{itemize}
\item[(i)] The set $\{ k_1^{\parallel}, \cdots k_r^{\parallel} \}$ is linearly
independent.
\item[(ii)] The set $\{ k_1^{\perp}, \cdots k_r^{\perp} \}$ is linearly independent.
\item[(iii)] The vector space $\Pi_T :=
\mbox{span}  \{ k_1^{\parallel}, \cdots,
k_r^{\parallel}, k_1^{\perp}, \cdots, k_r^{\perp} \}$ is $2r$-dimensional and 
$g^{\flat} |_{\Pi_T}$ has signature $\{ r,r\}$. Moreover, there
exists an orthonormal basis $\{e_1,  \cdots, e_{2r} \}$ of $\Pi_T$ with the
properties 
\begin{itemize}
\item[(a)]  $\mbox{span} \{ e_1, \cdots, e_r \} = \mbox{span}\{
 k_1^{\parallel}, \cdots k_r^{\parallel} \}$.
\item[(b)]  $\mbox{span} \{ e_{r+1}, \cdots, e_{2r} \} = \mbox{span}\{
 k_1^{\perp}, \cdots k_r^{\perp} \}$.
\item[(c)]  $\Pi = \mbox{span} \{ e_1+ e_{r+1}, \cdots, e_r + e_{2r} \}$.
\end{itemize}

\item[(iv)] A vector $v \in V$ is orthogonal to $\Pi$ if and only
if there exists $\overline{v} \in   \Pi_T^{\perp}$ such that 
$v- \overline{v} \in \Pi$.
\end{itemize}
\end{lemma}

\begin{proof}
Item (i) uses a similar argument as in the previous proof.
Indeed, if $\{ k_1^{\parallel}, \cdots k_r^{\parallel} \}$ were
linearly independent, say $k_r^{\parallel} = 
\sum_{i=1}^{r-1} k_i^{\parallel}
$, by the argument in the proof of 
\ref{TotallyDeg} we would have that 
$k_r^{\perp} - \sum_{i=1}^{r-1} a_i
k_i^{\perp}$ has zero norm and belongs to a space (namely $T^{\perp}$)
 where the metric is positive or negative definite.  Hence, this vector is
zero and we conclude that $k_r = \sum_{i=1}^{r-1} k_i^{\parallel}$, which 
is a contradiction. The proof of item (ii) follows the same steps.

To show (iii), we first note that  the orthogonal decomposition
$V = T \oplus T^{\perp}$ implies that 
$\Pi_T := \mbox{span} \{ k_1^{\parallel}, \cdots,k_r^{\parallel} \}
 \oplus \mbox{span} \{ k_1^{\perp}, \cdots,k_r^{\perp} \}$. The dimension
of $\Pi_T$ is $2r$ as a consequence of (i) and (ii) and  the
signature of $g^{\flat}|_{\Pi_T}$ is clearly $\{r,r\}$ 
because  $g^{\flat} |_{T}$ and $g^{\flat}|_{T^{\perp}}$  are positive and negative definite,
or viceversa.

Given that  $\mbox{span} 
\{ k_1^{\parallel}, \cdots, k_r^{\parallel} \}$ 
endowed with the restriction of $g^{\flat}$
defines a riemannian vector space, we can apply the Gram-Schmidt
orthonormalization procedure to define an adapted orthonormal basis
$\{ e_1, \cdots, e_r\}$. It follows that 
$e_i
= \sum_{j=1}^r a_i^j k_j^{\parallel}$. 
%with $a_i^i>0$ (no sum in $i$). 
%Note also the matrix $(a_i^j)$ 
%defined by this
%change of basis is lower triangular. 
We claim that the vectors $e_{r+i} :=
 \sum_{j=1}^r a_i^j k_j^{\perp}$, $i=1,\cdots,p$, 
define an orthonormal basis of 
$\mbox{span} \{ k_1^{\perp}, \cdots, k_r^{\perp} \}$. Indeed,
the conditions $\la k_i,k_j \ra =0$ are equivalent to
$\la k_i^{\parallel}, k_j^{\parallel} \ra = - \la k_i^{\perp}, k_j^{\perp} \ra$, and 
then
\begin{align*}
\la e_{r+i}, e_{r+j} \ra & =
\Big \la \sum_{l=1}^r a_i^l k^{\perp}_l,
\sum_{m=1}^r a_i^m k^{\perp}_m \Big \ra = 
\sum_{l=1}^r \sum_{m=1}^r a_i^l a_j^m  \la k_l^{\perp}, k_m^{\perp} \ra \\
& = - 
\sum_{l=1}^r \sum_{m=1}^r a_i^l a_j^m  \la k_l^{\parallel}, k_m^{\parallel} \ra
= - \Big \la \sum_{l=1}^r a_i^l k^{\parallel}_l,
\sum_{m=1}^r a_i^m k^{\parallel}_m \Big \ra =  - \la e_i, e_j \ra = - \sigma 
\delta_{ij}
\end{align*}
where $\sigma := g^{\flat}(e_1,e_1)$. If we denote by
$(b^{i}_j)$ the inverse matrix of 
$(a^i_j)$ it follows that
\begin{align*}
k_i = k_i^{\parallel}+ k_i^{\perp} =  
\sum_{j=1}^r b_i^j   e_i
+   \sum_{j=1}^r b_i^j   e_{r+i}
= 
\sum_{j=1}^r b_i^j  \left ( e_i + e_{r+i} \right ),
\end{align*}
which in particular implies that $\Pi = 
\mbox{span} \{ k_1, \cdots, k_r \}$ is also 
$ \Pi = \mbox{span} \{ e_i + e_{r+i} \}$. This proves (iii).

To establish (iv),  decompose $v = 
\sum_{i=1}^r a_i e_i + b_i e_{r+i} + \overline{v}$ according to the
orthogonal decomposition $V = \Pi_T \oplus
\Pi_T^{\perp}$. The condition that $v$ is orthogonal to all
$\{ k_i \}$, i.e. to all $\{ e_i + e_{r+i} \}$ imposes
$\sigma ( a_i - b_i )=0$ and we conclude
\begin{align*}
v - \overline{v} = \sum_{i=1}^r a_i (e_i + e_{r+i} ) \in \Pi
\end{align*}
as claimed.
\end{proof}

\begin{lemma}
\label{Ker}
Let $F^{\sharp}$ be a skew-symmetric endomorphism in an $n$-dimensional 
vector space
$V$ endowed with an inner product $g^{\flat}$ of signature $\{ p, q\}$. 
If $\mbox{dim} (\mbox{Ker} (F^{\sharp} ) \geq n-1$ then $F^{\sharp}=0$ and
conversely.
\end{lemma}

\begin{proof}
If $\mbox{dim} (\mbox{Ker} (F^{\sharp} ) = n$ there is nothing to prove, so let us assume that the kernel has dimension $n-1$, i.e. $\mbox{rank} (F^{\sharp} )
= 1$ or, equivalently that there exists a non-zero vector 
$k \in V$ such that
$F^{\sharp} (u) = \a(u) k$, for all $u \in v$. By linearity $\a(u)$ is a one-form,
hence a continuous linear map. By skew-symmetry
\begin{align*}
0 = \la u, F^{\sharp}(u) \ra = \a(u) \la u,k \ra. 
\end{align*}
Thus, $a(u) =0 $ on all vectors not lying in $k^{\perp} := \{ v \in V,
\la k,v \ra )=0 \}$. The inner product being non-degenerate,  $k^{\perp}$ 
has dimension at most $n-1$ and hence its complementary is dense in 
$V$. The one-form $\a$ vanishes on this set and hence everywhere
by continuity.
\end{proof}

\begin{lemma}
\label{FsqLemma}
Let $F^{\sharp}$ be a skew-symmetric endomorphism in a  vector space
$V$ endowed with an inner product $g^{\flat}$ of signature $\{ p, q\}$. 
The condition  $F^{\sharp} \circ F^{\sharp} =0$. is equivalent to
\begin{itemize}
\item[(i)] If $p=1$ or $q=1$:  $F^{\sharp} =0$
\item[(2)] If $p=2$ and $q \geq 2$:  $F^{\sharp} = k \otimes \bm{\ell} - 
\ell \otimes k$, where $\{ k,\ell\}$ is a basis of a two-dimensional
totally degenerate linear subspace.
\end{itemize}
\end{lemma}

\begin{proof}
By skew-symmetry 
\begin{align*}
\la F^{\sharp} \circ F^{\sharp}(u), v\ra = - \la F^{\sharp}(u), F^{\sharp}(v) \ra,
\quad \quad \forall u,v \in V
\end{align*}
so the condition $F^{\sharp} \circ F^{\sharp} =0$ is equivalent to
the linear space $\Pi := \mbox{Image} (F^{\sharp})$ being totally degenerate.
By Lemma \ref{TotallyDeg} the dimension
of $\Pi$ is at most one when $p=1$ and at most two when $p=2, q\geq 2$.
$\mbox{rank} (F^{\sharp}) \leq 1$ is equivalent to
$\mbox{dim} (\mbox{ker} (F^{\sharp} ) \geq n-1$ and by Lemma \ref{Ker} 
this happens if and only if  $F^{\sharp} =0$. 

It remains to consider the case $p=2, q \geq 2$ with $\Pi$ two-dimensional.
Let $\{k_1,k_2\}$ be a basis and fix a two
two-dimensional linear subspace $T \in V$ with negative definite
induced inner product. As before  the orthogonal 
decomposition $V = T \oplus T^{\perp}$ allows us to write
$v = v^{\parallel} + v^{\perp}$ for any vector  $v$.
By Lemma \ref{Ortogonality} we know that $\Pi = \mbox{span} \{ e_1 + e_{3}, e_2+ e_4 \}$ where $\{ e_i\}$ is an orthonormal basis of 
$\Pi_T := \mbox{span} \{ k_1^{\parallel}, k_2^{\parallel} \} \oplus 
\mbox{span} \{  k_1^{\perp}, k_2^{\perp} \}$ which is adapted to the direct sum decomposition
$T \oplus T^{\perp}$. As $\Pi$ has been defined as  the image of $F^{\sharp}$, 
there exist two non-zero one-forms
$\a,\b$ such that
\begin{align}
F^{\sharp} (u) = \a(u) (e_1 + e_3) + \b(u) (e_2+ e_4).
\label{Fsharp1}
\end{align}
By skew symmetry, any $u \in \Pi_T^{\perp}$ must satisfy $F^{\sharp}(u) \in
\Pi_T^{\perp}$. Thus $\a |_{\Pi_T^{\perp}} = \b |_{\Pi_T^{\perp}}=0$. Also by
skew symmetry $F^{\sharp}(e_i)$ ($i=1,2,3,4$) is perpendicular to $e_i$, so
$\a(e_1)= \a(a_3) = \b(e_2)=\b(e_4)=0$ and, in addition,
\begin{align*}
0 = \la F^{\sharp} (e_1), e_2  \ra +
\la e_1, F^{\sharp}(e_2) \ra = \b(e_1) - \a(e_2)  \quad \quad \Longleftrightarrow
\quad \quad \b(e_1) = \a(e_2).
\end{align*}
Applying $F^{\sharp}$ to (\ref{Fsharp1})  we find 
\begin{align*}
0 & = F^{\sharp} \circ F^{\sharp} (u) = 
\a(u) ( \b(e_1) + \b(e_3) )    (e_2 + e_4)
+ \b(u) ( \a(e_2) + \a(e_4)) (e_1 + e_3) \\
&  \Longleftrightarrow \quad \quad
\b(e_1) + \b(e_3)=0 \quad \mbox{and} \quad
\a(e_2) + \a(e_4)=0,
\end{align*}
where we used the fact that neither $\a$ nor $\b$ can vanish identically
(otherwise the rank of $F^{\sharp}$ would not be two).
Putting things together,  there exist a
non-zero constant $\alpha := \a(e_4)$ such that
$\a = \alpha ( \bm{e_2} + \bm{e_4}  )$ and
$\b = - \alpha ( \bm{e_1} + \bm{e_3} )$. We conclude that
\begin{align*}
F^{\sharp} = \alpha  ( \bm{e_2} + \bm{e_4} ) \otimes (e_1 + e_3 )
- \alpha  ( \bm{e_1} + \bm{e_3} ) \otimes (e_2 + e_4),
\end{align*}
which is $F^{\sharp} = k \otimes \bm{\ell} - \ell  \otimes \bm{k}$ after
defining $k = e_1 + e_3$ and $\ell = \alpha (e_2 + e_4)$. Since
$\{ k, \ell\}$ is a basis of $\Pi$ the lemma is proved.
\end{proof}

\begin{lemma}
\label{rank1}
Let $F^{\sharp}$ be a skew-symmetric endomorphism in a vector space $V$ endowed
with an inner product $g^{\flat}$ of signature $\{p,q\}$ with either
$p$ or $q$ different from zero.  Assume that
$F^{\sharp} \circ F^{\sharp} = \mu k \otimes \bm{k}$ with $k \in V$ non-zero
and null and $\mu \neq 0$. Suppose, moreover, that
\begin{itemize}
\item[(i)] $V$ is of Lorentzian signature, or
\item[(ii)] $\mbox{Image} (F^{\sharp} |_{\la k \ra^{\perp}}) \subset   \mbox{span}( k) $.
\end{itemize}
Then, and only then, there exists $w \in V$, linearly independent
and orthogonal to $k$ such that
\begin{align}
  F^{\sharp} = k \otimes \bm{w} - w \otimes \bm{k}, \quad
\quad \mu = - \la w ,w \ra.
\label{formF}
\end{align}
\end{lemma}

\begin{proof}
Let $\ell \in V$ be transverse to the codimension-one vector
subspace  $\la k \ra^{\perp}$. Since $\la k , \ell\ra \neq 0$ we may (after scaling 
$\ell$ if necessary) assume that  $\la \ell, k \ra = 1$. Define
$w := -F^{\sharp} (\ell)$ and observe
\begin{align*}
F^{\sharp} (w) = - F^{\sharp} \circ F^{\sharp} (\ell)= - \mu k. 
\end{align*}
Thus, for all $u \in V$,
\begin{align*}
\la k, F^{\sharp}(u) \ra & = - \frac{1}{\mu} \la F^{\sharp} (w),
F^{\sharp}(u) \ra = \frac{1}{\mu} \la w, F^{\sharp} \circ F^{\sharp} (u)
\ra  = \la w, k \ra \la k, u \ra,  \\
\la F^{\sharp} (u), F^{\sharp}(u) \ra
& = - \la u , F^{\sharp} \circ F^{\sharp} (u) \ra = - \mu \la k,u \ra^2.
\end{align*}
In particular, for $u \in \la k \ra^{\perp}$, $F^{\sharp} (u)$ is null and
orthogonal to $k$. In Lorentzian signature, this can only occur if
and only if $F^{\sharp} (u)$ is proportional to $k$ and we fall into
case (ii). We may thus assume (ii) irrespectively of the signature.

We first prove $F^{\sharp} (k)=0$. Indeed, under (ii), there is $\nu \in
\mathbb{R}$ such that $F^{\sharp}(k) = \nu k$. Since
\begin{align*}
F^{\sharp} \circ F^{\sharp}(k) = \mu k \la k, k \ra =0
\quad \quad \Longrightarrow \quad \quad F^{\sharp} (\nu k ) = \nu^2 k =0
\end{align*}
we conclude that $\nu$ must vanish, i.e. 
$F^{\sharp}(k) = 0$. From $V = \mbox{span} ( \ell ) \oplus \la k \ra^{\perp}$, it follows
that $\mbox{Image} (F^{\sharp}) = \mbox{span} (k,w)$ and, moreover, that
$\{w, k\}$ are linearly independent
(otherwise $\mbox{Image} (F^{\sharp}) = \mbox{span} (k) \subset \mbox{Ker}
F^{\sharp}$ and $F^{\sharp} \circ F^{\sharp}$ would be zero, contradicting the 
assumptions). There
exists two one-forms $\a, \b \in V^{\star}$ such that
\begin{align*}
    F^{\sharp} (u) = \a(u)  k  + \b(u)  w, \quad \quad
\forall u \in V.
\end{align*}
Applying $F^{\sharp}$ yields 
\begin{align*}
\mu \la k, u \ra k = F^{\sharp} \circ F^{\sharp} (u) = 
F^{\sharp} ( \a(u)  k  + \b(u)  w ) = 
 - \b(u)  \mu k
\quad \quad \Longleftrightarrow \quad \quad \b = - \bm{k}.
\end{align*}
Skew symmetry then forces $\a =  \bm{w}$ so that
\begin{align*}
F^{\sharp} = k \otimes \bm{w} - w \otimes \bm{k}
\end{align*} 
and we still need to impose
\begin{align*}
- \mu k = F^{\sharp} (w) = k \la w, w \ra - w \la k, w \ra
\quad \quad \Longleftrightarrow \quad \quad
\la k, w \ra = 0, \,\, \la w, w \ra = -\mu 
\end{align*}
This proves the ``then'' part of the lemma. The ``only then'' is immediate
since an $F^{\sharp}$ given by (\ref{formF}) with $k$ null and $w$ perpendicular
to $k$ immediately satisfies 
$F^{\sharp} \circ F^{\sharp} = - \la w, w \ra k \otimes \bm{k}$
\end{proof}

\begin{lemma}
\label{Fkperp}
Let $F^{\sharp}$ be a skew-symmetric endomorphism in a vector space $V$ endowed
with an inner product $g^{\flat}$ of signature $\{p,q\}$ with either
$p$ or $q$ different from zero.  Assume that there exists $k \in V$, non-zero and
satisfying $\la k ,k \ra =0$ such that
$F^{\sharp} |_{\la k \ra^{\perp}}$ takes values in $\mbox{span}( k )$.
Then there exists $v \in V$ such that
\begin{align}
  F^{\sharp} = k \otimes \bm{v} - v \otimes \bm{k}
\label{Fkv}
\end{align}
\end{lemma}
\begin{proof}
We follows a similar path as in the proof of Lemma (\ref{rank1}).
Let $\ell \in  V$ be a vector transverse to $\la k \ra^{\perp}$
and define $s := F^{\sharp} (\ell)$. 
Since $\mbox{span} ( \ell ) \oplus \la k \ra^{\perp} = V$, the hypothesis
of the lemma implies that $\mbox{Image} (F^{\sharp}) =  
\mbox{span} (k,s)$.  If $s$ is proportional to $k$, then
the rank of $F^{\sharp}$ is at most one
and Lemma \ref{Ker} implies that $F^{\sharp}=0$ 
which is of the form  (\ref{Fkv}) with $v$ proportional to $k$.
Thus, we may assume that  $v$ and $k$ are linearly independent.  
There exists two one-forms $\a, \b$ in the dual space
$V^{\star}$ such that
\begin{align*}
F^{\sharp} = k \otimes \a + s \otimes \b.
\end{align*}
The condition $F^{\sharp}(u)$ proportional to $k$
for all $u \in \la k \ra ^{\perp}$ requires 
$\b |_{\la k \ra^{\perp}}  =0$ (here we use that $k,s$ are linearly independent), or
equivalently
$b = - c k$ for some non-zero constant  $c$
(if it were zero, then $s = F^{\sharp} (\ell)$ would not be linearly
independent of $k$). By skew-symmetry, we conclude
\begin{align*}
F^{\sharp} = c \left (  k \otimes \bm{s} - s \otimes \bm{k} \right )
\end{align*}
which is (\ref{Fkv}) after defining $v =  c s$. 
\end{proof}

All the ingredients to prove the theorem are already in place.

\begin{proof}[Proof of Theorem \ref{degenerate}] 
Let $\zeta$ be the Killing vector of $\ads_a$ for which $\H$ is
a degenerate Killing horizon. We view $\H$ as a codimension submanifold of
$\M_{\epsilon}^{n+2}$ and we know there is  a skew-symmetric $F^{\sharp} : 
\M_{\epsilon}^{n+2} \longrightarrow \M_{\epsilon}^{n+2}$ such that $\zeta = \zeta_{F^{\sharp}} 
|_{\Sigma_a}$.  Let $\lambda := - g_{\tinyads_a}(\zeta,\zeta)$ be (minus)  
the square norm of $\zeta$ in the (A)-de Sitter  space.  By  definition 
of degenerate Killing horizon,
$\lambda |_x =0$ and $\mbox{grad} (\lambda) |_x =0$ at all
points $x \in \H$. The function
$\lambda$ is the restriction to $\Sigma_a$ of (minus) the square norm
of $\zeta_{F^{\sharp}}$, which is
\begin{align*}
\widetilde{\lambda} (x) 
:= - \la F^{\sharp} (x), F^{\sharp} (x) \ra = \la x, F^{\sharp} \circ F^{\sharp} (x) \ra.
\end{align*}
The gradient $\mbox{grad} (\lambda) |_x$ vanishes if and and only if
$\mbox{grad} (\widetilde{\lambda}) (x)$ is normal to $\Sigma_a$ at $x$. 
$\Sigma_a$ admits $\la x, x \ra  - \epsilon a^2$ as
defining  function, so the normal vector to this hypersurface
is $n = x$. The gradient is
$\mbox{grad} (\widetilde \lambda) = 2 F^{\sharp} \circ F^{\sharp} (x)$, so at every
point $x \in \H$, there must exist a real number $b|_x$ such that
\begin{align*}
F^{\sharp} \circ F^{\sharp} (x) = b|_x x.
\end{align*}
In addition it must be that $\widetilde{\lambda} |_x =0$, i.e.
$\la x, F^{\sharp} \circ F^{\sharp} (x) \ra = 0 $ and we conclude, taking into 
account that 
$x$ is non-null,
\begin{align*}
F^{\sharp} \circ F^{\sharp} (x) = 0, \quad \quad \forall x \in \H.
\end{align*}
This condition is linear in $x$ and $\H$ is everywhere transversal to 
the rays $\sigma x$, $\sigma \in \mathbb{R}$. In addition, the dimension
of $\H$ is $n-1$. Thus, the kernel
$F^{\sharp} \circ F^{\sharp}$ must be at least of dimension $n$ (equivalently,
the rank of $F^{\sharp} \circ F^{\sharp}$ is at most one). We now distinguish
two cases (a) 
$\mbox{rank} (F^{\sharp} \circ F^{\sharp} ) =1$ or 
(b) $\mbox{rank} (F^{\sharp} \circ F^{\sharp} ) =0$.

We start with  (a). Let $k$ be a generator of 
$\mbox{Image} (F^{\sharp} \circ F^{\sharp})$. Since $F^{\sharp} \circ  F^{\sharp}$
is symmetric, there is $\mu \in \mathbb{R} \setminus \{ 0 \}$ such that
\begin{align}
F^{\sharp} \circ  F^{\sharp} = \mu k \otimes \bm{k}. 
\label{Fsq}
\end{align}
The kernel of $F^{\sharp}
\circ F^{\sharp}$ is therefore $\Pi_k := \la k \ra^{\perp}$ which is a
codimension one hyperplane of $\M_{\epsilon}^{n+2}$. As shown above, a necessary condition for $x \in \Sigma_a$ to lie in a Killing horizon of $\zeta_{F^{\sharp}}$
is that  $ x \in \Pi_k$. Since $\Pi_k$ is transverse to $\Sigma_a$,
the intersection $\Sigma_a \cap \Pi_k$ is a smooth codimension one submanifold
in $\Sigma_a$, and $\H$ must be an open subset thereof. 
For any fixed  $x  \in \H$,  the tangent plane
 $T_x \H$ is
a codimension-two  vector subspace of $\M_{\epsilon}^{n+2}$ (we make the usual
identification of $\M_{\epsilon}^{n+2}$ and $T_{p} \M_{\epsilon}^{n+2}$ induced by the affine
structure).  Moreover
$T_x \H$ satisfies
\begin{align}
\label{inclusion}
T_x \H \subset \Pi_k, \quad \quad
T_x \H \subset \la x \ra^{\perp},
\end{align}
the first because $\H$ is a hypersurface of the linear
space $\Pi_k$ 
and the second because $T_x \Sigma_a 
= \la x \ra ^{\perp}$.  
%The geometry of null hyperplanes in Lorentzian spaces
%forces the existence of a set of $n-1$ vectors $\ell, e_A \in T_x \H_{\zeta}$
%where $A \in \mathbb{N}$ takes values in  $1 \leq A  \leq n-2$ (note that we ar%e
%including the case $n=2$, since then $A$ simply takes no values) 
%and it holds
%\begin{align*}
%\la \ell, \ell \ra =0, \quad \quad \la \ell, e_A \ra =0,
%\quad \quad \gamma_{AB} := \la e_A, e_B ra \mbox{\, positive definite} 
%\end{align*}
%Any vector $x \in  \M_{\epsilon}$  has non-zero norm, so it
%defines an orthogonal direct sum decomposition $T_x \M_{\epsilon}  = \la x \ra 
%\oplus T_x \Sigma_a$.
The property $x \in \H \subset \Pi_{k}$, i.e. $x$ normal to $k$ also says says
that $k$ is  tangent to $T_x \Sigma_a$.
By (\ref{inclusion}) 
$k$ is a normal  vector of the  null hyperplane $T_x \H$ 
within the Lorentzian vector space $T_x \Sigma_a$. This can only occur
if $k$ has zero norm $\la k, k \ra =0$. Moreover 
$\H$ being a Killing horizon of $\zeta$ requires that 
the Killing vector at $x$ is proportional to $k$, i.e
\begin{align}
F^{\sharp} (x) = q|_x k, 
\label{Fx}
\end{align}
with $q|_x$  non-zero given that $\zeta|_x$ does not
vanish anywhere on its Killing horizon. 
Applying $F^{\sharp}$  to (\ref{Fx}) one finds
\begin{align}
F^{\sharp} \circ F^{\sharp} (x) = q|_x F^{\sharp} (k) = \mu k \la k, x \ra 
=0   \quad \quad
\Longrightarrow \quad \quad F^{\sharp} (k) = 0 
\end{align}
where in the second equality we used (\ref{Fsq}).  Take any vector
$s$ tangent to $T_x \H$. Given that $\la s,x \ra = \la s, k \ra =0$,
skew symmetry implies
\begin{align*}
\la F^{\sharp} (s),  k  \ra = \la F^{\sharp} (s), x \ra = 0,
\end{align*}
so $F^{\sharp}(s)$ is also tangent to $T_x \H$. Moreover,
\begin{align*}
0 = - \mu \la k,s \ra^2  =  - \la F^{\sharp} \circ F^{\sharp}(s), s \ra
= \la F^{\sharp}(s), F^{\sharp}(s) \ra
\end{align*}
so $F^{\sharp}(s)$ has zero norm. It must therefore be that 
$F^{\sharp}  (s) \in \mbox{span}(k)$ for all vectors in $T_x \H$. Using  $\la x
\ra \oplus T_x \H_{\zeta} = \la k \ra^{\perp}$  we conclude that
$F^{\sharp}$ maps $\la k \ra^{\perp}$ into $\mbox{span}( k)$. Thus, by
Lemma \ref{rank1} 
there exists $w \in \M_{\epsilon}^{n+2}$ linearly independent to $k$, orthogonal to $k$
and satisfying $\la w, w \ra = -\mu \neq 0$ such that
$F^{\sharp} = k \otimes \bm{w} - w \otimes \bm{k}$. This proves the ``if'' part
of the theorem in case (a).  For the converse,
we check that, given such $k$ and $w$, the Killing vector $\zeta_{F^{\sharp}}$ admits
as degenerate Killing horizon the hypersurface
\begin{align}
\H_{\zeta} := (\Sigma_a \cap \la k \ra^{\perp})  \setminus \{ \la w, x \ra = 0 \}.
\label{defHconverse}
\end{align}
Indeed, the square norm  of $\zeta := \zeta_{F^{\sharp}} |_{\Sigma_a}$
is $\lambda = - g_{\tinyads_a} (\zeta, \zeta) = - \mu \la k ,x \ra^2$ which vanishes
on $\Sigma_a \cap \la k \ra ^{\perp}$. This is a smooth null embedded
hypersurface of $\ads_a$  and $\zeta$ restricted to this hypersurface takes the form  $\zeta = k \la w, x \ra$, so it is  tangent, null, and non-zero 
exactly  on $\H_{\zeta}$. Moreover, this Killing 
horizon is degenerate because   $d \lambda \eqHzeta 0$. This concludes the 
proof of the theorem in case (a).

We now consider case (b), i.e. we assume 
$F^{\sharp} \circ F^{\sharp}  =0$.
By  Lemma \ref{FsqLemma} (item (i))  we see  that this can
only happen in the anti-de Sitter case (i.e. $\epsilon =-1$) and in
dimension  $n \geq 2$. Applying item (ii) in the same lemma, there is $\{\ell_0,k_0\}$
basis of a two-dimensional totally degenerate linear subspace 
such that $F^{\sharp} = k_0 \otimes \bm{\ell_0} - \ell_0 \otimes \bm{k_0}$.
The Killing vector $\zeta$ is null everywhere 
which opens up the possibility
that $\antids$ is foliated by Killing horizons of $\zeta$. To confirm this we need
to check first that $\zeta$ vanishes nowhere. Assume, on the contrary that
there is $x \in \Sigma_a$ where $\zeta |_x =0$. Then
\begin{align*}
0 =  \zeta|_x = \zeta_{F^{\sharp}} |_x = F^{\sharp} (x) = k_0 \la \ell_0,x \ra  
- \ell_0 \la k_0, x \ra.
\end{align*}
By linear independence this can only happen if $\la \ell_0,x \ra =
\la k_0, x \ra =0$.   Applying item (iv) of Lemma \ref{Ortogonality} we 
conclude that $x = \overline{x} + a_1 k_0 + a_2 \ell_0$ for some constants $a_1,a_2$.
Moreover $x \in \Sigma_a$ so
\begin{align*}
- a^2 = \la x, x \ra = \big \la \overline{x} + a_1 k_0 + a_2 \ell_0, \overline{x}
+ a_1 k_0 + a_2 \ell_0 \big \ra = \la \overline{x}, \overline{x} \ra 
\end{align*}
which is impossible since $\overline{x}$ lies in a space with positive
definite inner product. 
Thus $\zeta$ has no zeros, and the Fr\"obenius theorem 
(see \cite{Lee}) implies immediately 
that $\antids_a$ is foliated the 
Killing prehorizons. We want to show that,
in fact, the foliation is by Killing horizons, i.e. that 
the leaves are embedded submanifolds (and identify them  explicitly).
Consider the collection of hyperplanes 
$\Pi_{\alpha} := \{ \cos \alpha  k_0  + \sin \alpha \, \ell_0 \}^{\perp} 
\subset \M_{\epsilon =-1}^{n+2}$, where $\alpha \in \mathbb{S}^1$.
and
 define 
$\H_{\alpha} := \Sigma_a \cap \Pi_{\alpha}$. 
The hyperplane $\Pi_{\alpha}$ is transverse to $\Sigma_a$. Indeed,
being both submanifolds of codimension one, they can fail to be transverse
only at points $x \in \antids_a$ where $T_x \Sigma_a = \Pi_{\alpha}$. 
This coincidence occurs iff
the corresponding normal vectors are parallel, i.e. iff $x 
= \nu ( \cos \alpha k_0 + \sin \alpha \ell_0)$ for some non-zero $\nu$.
But this immediately contradicts $\la x, x \ra = -a^2 \neq 0$.

Transversality of $\Pi_{\alpha}$ and $\Sigma_a$ implies that
$\H_{\alpha}$ is an embedded
submanifold of $\antids_a$. We claim that $\H_{\alpha}$ is a Killing
horizon of $\zeta$. Note first that the two vectors
\begin{align*}
k_1 &:= \cos \alpha k_0 + \sin \alpha \ell_0,  \\
\ell_1 &:= - \sin \alpha k_0 + \cos \alpha \ell_0
\end{align*}
are linearly independent (hence a basis of the totally degenerate
plane $\Pi$) and satisfy
\begin{align*}
F^{\sharp} := k_0 \otimes \bm{\ell_0} - \ell_0 \otimes \bm{k_0}
= k_1 \otimes \bm{\ell_1} - \ell_1  \otimes \bm{k_1}.
\end{align*}
Moreover by construction $k_1$ is tangent to $\Pi_{\alpha}$ (because
$k_1$ is orthogonal to itself).
At any point $x \in \H_{\alpha} \subset \Pi_{\alpha}$ the Killing vector $\zeta$ 
takes the form
\begin{align*}
\zeta|_x = k_1 \la \ell_1, x \ra - \ell_1 \la k_1,x \ra
= k_1 \la \ell_1, x \ra
\end{align*}
Thus $\zeta|_x$ is null, non-zero and tangent to $\H_{\alpha}$. Moreover $\H_{\alpha}$ a null hypersurface
of $\antids_a$ 
because 
$k_1$ is normal to $\H_{\alpha}$  (any vector $v \in T_p \H_\alpha$ must also belong to $\Pi_{\alpha}$, which requires
$\la v, k_1 \ra =0$). We conclude that 
$\H_{\alpha}$ is a Killing horizon of $\zeta$.  Note that 
$\Pi_{\alpha} = \Pi_{\alpha + \pi}$, so we may restrict $\alpha$ to lie in
$(-\pi/2,\pi/2]$. We claim that the collection
of such $\{ \H_{\alpha} \}$,  defines a foliation of $\antids_a$ by
embedded null hypersurfaces. Indeed, assume 
$\alpha \neq \beta$ then $\Pi_{\alpha} \cap \Pi_{\beta}$ is the collection
of points $x_0 \in \M_{\epsilon =-1}^{n+2}$ orthogonal to both $k_0$ and $\ell_0$,
which are characterized in item (iv) of Lemma \ref{Ortogonality}. 
We have shown above that none of of these points belongs to $\Sigma_a$.
Thus $\H_{\alpha} \cap \H_{\beta} = \varnothing$. The collection
$\{ \H_{\alpha}\}$ defines a foliation provided for any $x \in \Sigma_a$,
there is $\alpha \in (-\pi/2,\pi/2]$  such that $x \in \H_{\alpha}$. But this
is clear because  the union $\bigcup_{\alpha \in \mathbb{S}^1} \Pi_{\alpha} 
= \M_{\epsilon}^{n+2}$, since for any
$x \in \M_{\epsilon=-1}^{n+2}$, the equation 
\begin{align}
\cos \alpha \la x, k_0 \ra + \sin \alpha \la x, \ell_0 \ra =0 
\label{eqalpha}
\end{align}
always admits solutions for $\alpha$ in this interval. 

We can now finish the proof of the theorem in case
(b). Let  $\H$ be the degenerate Killing horizon in the statement of the theorem, $\zeta$ any Killing vector of $\antids_a$ for which
either $\H$ or an open and dense subset  thereof is a degenerate Killing horizon of $\zeta$,
and assume that $\zeta = \zeta_{F^{\sharp}}$ with $F^{\sharp} \circ F^{\sharp} =0$.
Let $k_0$, $\ell_0$ be such that 
$F^{\sharp} = k_0 \otimes \bm{\ell_0} - \ell_0 \otimes \bm{k_0}$. 
Fix $x \in \H$ and solve (\ref{eqalpha}). Since 
$x \in \H \subset \Sigma_a$ not both $\la x, k_0\ra$ and $\la x,\ell_0\ra =0$
vanish, and the equation admits precisely  one  solution $\alpha_0
\in (\pi/2,\pi/2]$. The hypersurface
$\H_{\alpha_0} := \Pi_{\alpha_0} \cap \Sigma_a$ is a maximal Killing
horizon of $\zeta$. Thus $\H$ is a subset of $\H_{\alpha_0}$. 
Setting $k = k_1$ and $w = \ell_1$ the direct part of the theorem follows.
The converse  is clear from the results above.
%admits  solution for $\alpha$ when at least one
%of $\la x, k \ra $, $\la x, \ell \ra$ is non-zero, and infinite
%solutions when both are zero (note that in the latter case $x$ never belongs
%to $\Sigma_a$, as shown above). 
\end{proof}

\end{document}